\documentclass[10pt,twocolumn,twoside]{IEEEtran}     

\IEEEoverridecommandlockouts                              
\usepackage{latexsym,epsfig,amssymb}
\usepackage{showidx}
\usepackage{tikz}
\usepackage[T1]{fontenc}
\usepackage{amsmath,mathtools}
\usepackage{cite}

\usepackage{latexsym,epsfig,subfigure}
\usepackage{verbatim}
\usepgflibrary{arrows}
\usepackage{color}
\usepackage{graphicx}
\usepackage{arydshln}

\tikzset{>=latex}

\tikzset{%
block/.style    = {draw, thick, rectangle, minimum height = 3em,
    minimum width = 3em},
  block1/.style    = {draw, thick, rectangle, minimum height = 1.7em,
    minimum width = 1.7em,fill=gray!70},
      block2/.style    = {draw, thick, rectangle, minimum height = 1.7em,
    minimum width = 1.7em,fill=gray!30},
  sum/.style      = {draw, circle, node distance = 1.8cm}, 
}

\title{Dynamics--Decoupling  Control  for Strings \\ of Heterogenous Nonlinear Autonomous Agents
}

\author{\c{S}erban Sab\u{a}u$^\sharp$, \thanks{$^\sharp$\c{S}erban Sab\u{a}u
is with the Electrical and Computer Engineering Dept.,
Stevens Institute of Technology, Hoboken, New Jersey, U.S.A.
email: {\tt ssabau@stevens.edu}}
Irinel--Constantin Mor\u{a}rescu$^\star$\thanks{$^\star$ I.-C. Mor\u{a}rescu is with Universit\'{e} de Lorraine, CRAN, UMR 7039 and CNRS, CRAN, UMR 7039, 2 Avenue de la For\^{e}t de Haye, 54506 Vand\oe uvre-l\`{e}s-Nancy, France. email: {\tt constantin.morarescu@univ-lorraine.fr}}, Lucian Bu\c{s}oniu$^\ddagger$\thanks{$^\ddagger$ L.~Bu\c{s}oniu is with Technical University of Cluj--Napoca, Romania. email:{\tt lucian@busoniu.net}}  \\
and Ali Jadbabaie$^\dagger$\thanks{$^\dagger$ Ali Jadbabaie is with the Institute for Data Systems and Society and the Laboratory for Information and Decision Systems, Massachusetts Institute of Technology. email: {\tt jadbabai@mit.edu.}}\thanks{The work of I.C. Mor\u{a}rescu was partially funded by ANR project Computation Aware Control Systems (COMPACS) and the project NETASSIST funded by the Agence Universitaire de la Francophonie (AUF) and the Romanian Institute for Atomic Physics (IFA)}}

\newtheorem{theorem}{Theorem}[section]
\newtheorem{rem}[theorem]{Remark}

\newtheorem{lem}[theorem]{Lemma} 
\newtheorem{prop}[theorem]{Proposition}
\newtheorem{defn}[theorem]{Definition}

\newtheorem{assumption}[theorem]{Assumption}

\newcommand{\ba}{\left[ \begin{array}}
\newcommand{\baa}{\begin{array}}
\newcommand{\ea}{\end{array} \right]}
\newcommand{\eaa}{\end{array}}

\newcommand{\be}{\begin{equation}}
\newcommand{\ee}{\end{equation}}
\newcommand{\bb}{\begin{equation}\label}

\newcommand{\figref}[1]{Figure~\ref{#1}}

\renewcommand{\tilde}{\widetilde}

\newcommand{\FF}{{{\rm I \kern -0.2em R}}}
\newcommand{\RR}{{{\rm I \kern -0.2em R}}}

\newcommand{\CC}{{{\mbox{\rm \hspace*{0.05ex}
\rule[.18ex]{.18ex}{1.24ex} \kern -.65em C}}}} 

\newcommand{\bea}{\begin{eqnarray}}
\newcommand{\eea}{\end{eqnarray}}

\newcommand{\R}{{\bf R}}

\begin{document}
\maketitle

\begin{abstract}  We introduce a distributed control architecture for a class of heterogeneous, nonlinear dynamical agents moving in the ``string'' formation, while guaranteeing trajectory tracking, collision avoidance and the preservation of the formation's topology. Each autonomous agent uses  information and  relative measurements only with respect to its predecessor in the string. The performance of the scheme is independent of the number of agents in the network and also on the agent's relative position in the network. The scalability is a consequence of the ``decoupling''  of a certain bounded approximation of the closed--loop equations, which allows the regulation and controller design (at each agent) to be done individually, in a completely decentralized manner. A practical method for compensating communication induced delays is also presented. Numerical examples illustrate the effectiveness and the main features of the proposed approach.
\end{abstract}


\section{Introduction}


Practical algorithms for distributed control of dynamically coupled systems  are needed in many diverse applications raging from formation control of autonomous mobile agents \cite{Ali,PATRU}, synchronization of local clocks offsets or phase differences between (neighboring) coupled oscillators \cite{Mallada}  or synchronous generators in power networks, sensor networks, load balancing \cite{SASE}, distributed agreement algorithms, cooperative control of multi-robot systems \cite{CINCI}, opinion dynamics {\em etc.}
In the specific setting of  autonomous agents, the intricacies of dynamical coupling are not caused by the structure of the plant but rather by: {\em (i)} the structure of  the cost functional resulting from the definition of the regulated measurements ({\em e.g.} in formation control - the inter-agent spacing distances defining the topology of the formation)  and {\em (ii)} the coupling induced in the entailing feedback loop (with the distributed controller). The subsequent controller design problem is further complicated by the constraints  to be imposed 
on the sensing and communications radii\footnote{By communications radius we mean the number of hops associated to an agent in the communications graph. The communications graph designates for each agent the neighbors with which it is able to communicate.} of each agent.

The goal of the distributed control scheme is for all the agents to attain certain types of  global collaborative behavior, such as their outputs/states reaching agreement in a precisely defined metric. In modern control parlance this class of objectives have been dubbed consensus \cite{OPT, NOUA} and synchronization \cite{ZECE, UNSPE, DOISPE, TREISPE, PAISPE} problems.  In the existing literature there is no clear demarcation line for this terminology (consensus versus synchronization), therefore we refer to \cite[Section~1]{SAPTE}, \cite[Section~1]{Rick2016} for a useful discussion. It is worth mentioning that synchronization is a far more ambitious objective than classical reference tracking \cite[Chapter~5.4]{Kwakernaak}, not only due to the distributed nature of the problem but also because the reference signals are never explicitly available to all agents, while in certain scenarios the synchronization trajectory is not even assigned beforehand and therefore there is no explicit reference to be tracked \cite{Rick2016}.

 In this paper we deal with a group of heterogenous, nonlinear dynamical agents that solve a conventional agreement task (velocities matching in our case) but to which we append a specific set of constraints 
linking the individual states of two adjacent agents  (in our case their positions in space). It turns out that the inclusion of such constraints render existing methods in distributed synchronization ({\em e.g.} distributed output regulation) inapplicable directly. The constraints relate to the inter-spacing distance between two agents, defined as the difference of their individual positions in space. In distance-based formation control\footnote{When referring to the position in space of an autonomous agent within a group of agents, the position must be defined with respect to an inertial system of reference common to all agents, {\em e.g.} a Global Positioning System. Note that the methods described here 
do not employ positioning systems, relying exclusively on measurements of {\em relative distances} between agents (acquired using for example onboard lidars).} such constraints on the  inter-spacing distances between neighboring agents arise naturally when: {\em (i)} defining the formation's steady-state topology and  {\em (ii)} when framing the collision avoidance requirement, via restrictions on inter-agent distances in the transitory regime.  For illustrative simplicity we will only look at the string graph, while our method handles the heterogeneity of the formation and the inter-agent communications time delays, overlooking applications in the automotive industry. For consistency, we will refer to this setting as a synchronization problem.

 Existing results in distributed synchronization (or distributed agreement) for Linear and Time Invariant  ({\bf LTI}) dynamics rely on observer-based distributed controllers \cite{SAPTE, CINSPE, SAISPE}, while the results from \cite{SAPTISPE, OPTISPE, NOUASPE} remove the necessity of communicating the internal states of the local observers/sub-controllers and rely solely on the communication of the agents' outputs. The results \cite{2ZECI, 2ZECISIUNU, 2ZECISIDOI} for the heterogenous case, rely on the existence of a (virtual) exo-system  generating the reference trajectory, while \cite{2ZECISITREI} outlines the intrinsic connections between the set of possible agreement trajectories and the sharing of all agents of certain ``common dynamics''. The authors' recent results  in \cite{TAC2016} provide a solution in the more ambitious setting of a distributed $\mathcal{H}_2 / \mathcal{H}_\infty$  disturbances attenuation problem for the string graph, encompassing heterogenous agents and communications induced time delays. In this context, the current paper can be looked at as an extension to nonlinear control of the novel ideas for LTI dynamics from \cite{TAC2016}.

 Theoretical advancements for nonlinear agents  are in an incipient phase and have only been available more recently. The Lyapunov function approach in \cite{2ZECISISAPTE} is based on differential inequalities. The results in \cite{2ZECISIPATRU} pertaining to weakly minimum phase nonlinear agents are viable only under a passivity hypothesis, while those in \cite{2ZECISICINCI, 2ZECISISASE} pertain to globally Lipshitz-like conditions on the nonlinearities and leader-following networks. The reference \cite{2ZECISIOPT} brings forward a necessary condition but no controller synthesis procedure while in \cite{2ZECISINOUA} the agreement objective can only be set to a constant.
Very recent results applicable to more complicated nonlinear dynamics include feedforward schemes \cite{3ZECI} or are set up as cooperative output regulation problems {\em e.g.} \cite{3ZECISITREI, 3ZECISIPATRU, 3ZECISIUNU, 3ZECISIDOI} and the references within, among which \cite{3ZECISITREI, 3ZECISIPATRU} deal with leader-following networks.  A notable feature of  \cite{3ZECISIUNU, 3ZECISIDOI}  and especially \cite{Rick2016} is that the sub-controller corresponding to an agent can be designed independently to all other sub-controllers. The downside of \cite{3ZECISIUNU, 3ZECISIDOI} is the  requirement of  full state information exchange among agents, requirement entirely circumvented in \cite{Rick2016} in a generic setting.

\subsection{Motivation and Scope of Work}

 The references above deal with the standard setup in which agents must achieve agreement of certain, pre-specified variables from each agent's own state-space. For the class of problems treated in this paper, this takes the form of guaranteeing  {\em velocity matching} in the steady-state, irrespective of the velocity profile of the leader (whose trajectory represents the reference for the entire formation) and which is seen as an adversarial player. However, the problem statement is further complicated by the inclusion of constraints that impose a substantial ``coupling'' between individual state variables of distinct agents, where these states represent the positions of agents\footnote{See also footnote 2 on page one.}. These constraints cast on the relative distances between two neighboring agents render the  existing methods referred above inapplicable directly\footnote{Existing results  \cite{Baillieul, AndersonRigid, MorseRigid} on {\em graph rigidity} show how easy it is for these types of constraints to cause certain variations of this formulation of the distance-based formation control problem to become not well-posed. That happens when in an effort to preserve the topology of the formation, the controller encounters simultaneously conflicting  constraints.}. However, the constraints are needed in order to frame sufficient conditions for: {\em (i)} {\em collision avoidance} in the transitory regimes and {\em (ii)} {\em topology preservation} of the formation in steady-state ({\em i.e.} the interspacing distances between agents must converge asymptotically to certain pre-specified, constant values).

For high performance displacement-based formation control \cite[Section~6]{formsurv}, global positioning systems are not viable due to their relative large latencies and problematic reliability. The absence of a global coordinate system (combined with the fact that we avoid the use of accelerometers\footnote{Longitudinal accelerometers are notoriously unreliable for applications in the automotive industry.}) requires that the agents must rely only on real time measurements of {\em relative variables} with respect to their neighbors, with all the difficulties such schemes entail, including the fact that collision avoidance and topology preservation cannot be reduced to a cooperative, output regulation problem (as those referred above).

\subsection{Contributions of the Paper}

Our controller's architecture is borrowed from platoon control literature\footnote{The conceptual architecture behind such distributed control schemes  have been dubbed Cooperative Adaptive Cruise Control in the platooning control parlance \cite{Ploeg, TAC2016}.} and is conceptually different from the aforementioned methods ({\em e.g.} \cite{3ZECISITREI, 3ZECISIPATRU, Rick2016} and the references within). Unlike \cite{3ZECISITREI, 3ZECISIPATRU} it doesn't require exchange of internal states (plant internal states or controller states) among agents. In turn, each agent needs to transmit its {\em control action} only to its immediate follower in the string. Furthermore, the design of each sub-controller, (``local'' to an agent ) can be done in a completely independent manner - feature which is known to be especially challenging in distributed synchronization  (see \cite{Rick2016} and the references within for a comprehensive discussion in a related setting). Indeed, solely the knowledge of the dynamical model of the immediate predecessor is required for the local sub-controller at each agent, but once this is made available the regulation and controller design (at each agent) is done individually, in a completely decentralized manner.




Perhaps the most appealing feature of the proposed scheme is a particular dynamic ``decoupling''  of a certain bounded approximation of the closed--loop equations, entailing that individual, {\em local} analyses of the closed--loop stability at each agent will in turn guarantee the aggregated stability of the entire formation. This entails a  complete scalability with respect to: {\em (i)} the number of agents in the string and {\em (ii)} the same  performance irrespective of the relative position in formation (front or back of the string).

By comparison to our method, the main result in  \cite{SCL2012} is restricted to an undirected topology of the distributed controller, with stringent requirements  involving: {\em (i)} the transmission of the exact state of the leader to many agents in the formation ({\em virtual leaders}) and {\em (ii)} the necessity of high control gains (see the last paragraph in \cite[page~1]{Sabau2017} for a more detailed discussion).

Overall, our scheme improves on existing results in the following essential aspects:
\begin{enumerate}
\item The agent dynamics are permitted to be  heterogenous as long as they are nonlinear,  globally Lipschitz.

\item The agents achieve the synchronization of their velocities in the steady-state, while guaranteeing collision avoidance.

\item The scheme guarantees steady-state topology preservation. Very recent results \cite{Nader} are able to achieve this but only for identical single-integrators, exploiting an  adaptation of the Cucker-Smale type nonlinear controllers \cite{CuckerSmale2007}.  Collision avoidance is obtained in \cite{CuckerDong2011} for single-integrators but without topology preservation.


\item The distributed controller determines a ``dynamic decoupling'' of the closed loop, rendering the same performance independent of the number of agents or the relative position in formation.  (The Lyapunov function guaranteeing the closed-loop stability of the entire formation is actually the sum of ``local'' Lyapunov functions, proper to each agent. This decoupling is also the root cause of the feature stated at the next point.)

\item Completely independent regulation and controller design at each agent, under the sole requirement that each agent knows the dynamical model of its predecessor\footnote{This aspect is essential when dealing with merging/exiting of agents, since it allows only local reconfigurations (at the merging agent or at the follower of the exiting agent) without the need to reconfigure the control scheme for the entire formation.}.

\item We provide a simple, practical method for the efficient compensation of the notoriously detrimental (communications induced) time-delays \cite{Rick2014}, at the expense of a negligible loss in performance.
\end{enumerate}

\subsection{Paper Organization}
The paper is organized as follows: in Section~II we introduce the general framework and problem formulations. Section~III provides a preliminary description of the novel distributed control architecture introduced in this work along with a first glimpse at the closed-loop dynamics ``decoupling'' featured by the control scheme. Section~IV contains the main result as it delineates the guarantees for stability, velocity matching, collision avoidance and topology preservation. Finally, Section~V outlines a practical delays compensation mechanism while Section~VI provides an illustrative numerical example, worked out on an actual dynamical model for road vehicles.
\section{General Framework and Problem Statement}

The notation being used is fairly standard throughout the literature, for example the derivatives $\frac{d}{dt}{z}(t)$ with respect to the time variable are sometimes denoted by  $\dot{{ z}} (t)$. Also, throughout the paper it will become apparent from the context when the time argument $(t)$ is being omitted for the sake of brevity. The notation $a \overset{def}{=} b$ means that the left hand side quantity $a$ is defined to be the right hand side quantity $b$.

\begin{defn} The $\sigma$--norm  of a vector $x$ is defined as

\begin{equation}\label{nrm}
\| x \|_\sigma \overset{def}{=} \frac{1}{\sigma} \Big[ \sqrt{1+ \| x \|_2^2 }  -1 \Big]
\end{equation}
where $\sigma$ is a strictly positive constant. Note that (\ref{nrm}) is a class $\mathcal{K}_\infty$ function of $ \| x \|_2^2$ and is differentiable everywhere.
\end{defn}
\begin{defn} A set $\Omega$ is said to be {\em forward invariant} with respect to an equation, if any solution $x(t)$ of the equation satisfies:
$
x(0)\in \Omega \Longrightarrow x(t) \in \Omega, \: \forall t>0.
$
\end{defn}

\begin{defn} \label{apf} {Artificial Potential Function} {(APF)}.   The function $V_{k,k-1}(\cdot)$ is a class $C^1$, nonnegative, radially unbounded function of $\|z \|_\sigma$ 
satisfying the following properties:

{\bf (i)} $V_{k,k-1}(\|z \|_\sigma) \rightarrow \infty$ as $(\|z \|_\sigma) \rightarrow 0$,

{\bf (ii)} $V_{k,k-1}(\|z \|_\sigma)$ has a unique minimum, which is attained at $\|z \|_2=\delta_k$, with  $\delta_k$ being a  positive constant.
\end{defn}

\subsection{Distributed Trajectory Tracking in the String Formation}
We consider a {\em heterogeneous} group of $n+1$ agents ({\em e.g.} autonomous road vehicles) moving along the same (positive) direction of a roadway, with the origin at the starting point of the leader. The dynamical model for the agents, relating the control signal $u_k(t)$ of the $k$--th vehicle  to its position $y_k(t)$ is given by

\begin{subequations} \label{dacia}
\begin{equation} \label{nlin}
\dot{{ y}}_k(t)={ v}_k(t), \quad \dot{{ v}}_k(t)={ f_k}({ v}_k(t)) + { u}_k(t) \ ;
\end{equation}\vspace{-6mm}
\begin{equation} \label{initcond}
y_k(0)= -\sum_{j=0}^k \ell_j, \quad v_k(0)=0.
\end{equation}
\end{subequations}
where $v_k(t)$ is the instantaneous speed of the $k$--th agent, $u_k(t)$ is its command signal and $\ell_k$ is the initial interspacing distance between the $k$--th agent and its predecessor in the string. Throughout the sequel we will use the notation

\begin{equation} \label{IO}
y_k= G_k \star u_k
\end{equation}
to denote (especially for the graphical representations) the input--output operator $G_k$ of the dynamical system from (\ref{nlin}), with the initial conditions (\ref{initcond}).

\begin{assumption} \label{lider}
The index ``$0$'' is reserved for the {\em leader agent}, the first agent in the string. This situation leads to exactly $n$ inter-agent distances, which are part of the regulated measurements.   
\end{assumption}

 \begin{figure*}[!ht]
\hspace{-3mm}
\centering
\includegraphics[scale=0.85]{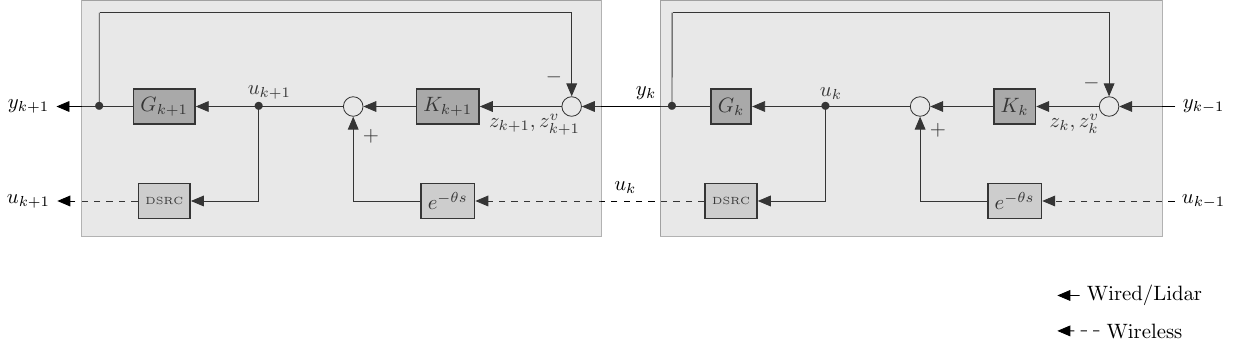}
\caption{Distributed Controller Implementation.} 
\label{f2}
\end{figure*}

In the rest of the paper it will become apparent from the context that we often omit the time argument $(t)$, for the sake of brevity.  Let us further define

\begin{equation} \label{z}
{{ z}}_k\overset{def}{=}{y}_{k-1}-{y}_k, \quad {{ z}}_k^v\overset{def}{=}{ v}_{k-1}-{ v}_k \quad \text{for} \quad 1 \leq k \leq n,
\end{equation}
to be the interspacing  and relative velocity error signals respectively (with respect to the predecessor in the string).  By differentiating the first equation in (\ref{z}) it follows that $\dot{{ z}}_k(t)={{ z}}_k^v(t)$, therefore implying that  constant interspacing errors (in steady state) are equivalent with zero relative velocity errors and also allowing to write the following  time evolution for the relative velocity error of the $k$--th vehicle

\begin{equation} \label{primo}
\dot{{ z}}_k^v={ f_{k-1}}({ v}_{k-1})-{ f_k}({ v}_k) + { u}_{k-1}-{ u}_k .
\end{equation}

\section{A Practical Distributed Control Architecture} \label{arqui}



After five decades of consistent academic efforts and hundreds of references on the subject, it turned out that  control of a string of mere double integrators might well be the epitome of the difficulties typical to distributed control, since it suffers from all pitfalls one might have expected from more general and complex dynamical networks, {\em e.g.} performance is in general dependant on the number of agents in the string and on their relative position in formation and is highly sensitive to communications delays.

We introduce a novel  control architecture featuring a highly beneficial ``decoupling'' property of the closed--loop dynamics, that resolves the troubling {\em nested} interdependencies of the regulated measurements.  We consider non--linear controllers built on the so-called Artificial Potential Functions {(APF)}, in particular we will look at control laws of the type

\begin{equation} \label{ours}
\begin{split}
{ u}_k &=  { u}_{k-1} + {\beta_k}  ({ v}_{k-1} -  { v}_k)   - \nabla_{{y}_k} V_{k,k-1}(\| {y}_{k-1} -{y}_k \|_\sigma)\\
&\quad \quad \quad \quad \quad \quad \quad \quad \quad -{f}_k({v}_k)+{f}_{k-1}({v}_k)
\end{split}\end{equation}
with $k \geq 1$, where each of the $V_{k,k-1}(\cdot)$ functions is an Artificial Potential Function \cite[Definition~7]{SCL2012}, with $\beta_k$ being a proportional gain to be designed for supplemental performance requirements. With the notation from (\ref{z}), the control policy (\ref{ours}) for the $k$--th agent becomes

\begin{equation}\label{oursbis}
\begin{split}
{ u}_k &=  { u}_{k-1} + {\beta_k}  z_k^v   - \nabla_{{y}_k} V_{k,k-1}(\| {z}_k \|_\sigma)\\
& \quad \quad \quad \quad \quad \quad \quad \quad \quad -{f}_k({v}_k)+{f}_{k-1}({v}_k).
\end{split}
\end{equation}
Note that the distributed control laws rely only on information locally available to each agent, since it can further be written as the sum of  the following two components: firstly, the control signal  $u_{k-1}(t)$ of the preceding agent, which is received onboard the $k$--th agent via wireless communications ({\em e.g.} digital radio) along with the  function ${f}_{k-1}(\cdot)$ characterizing the predecessor's  dynamical model. Secondly, the {\em local} component, which we denote with

\begin{equation} \label{FigureHelper}
u_{k}^\ell \overset{def}{=} {\beta_k} {z}_k^v   - \nabla_{{y}_k} V_{k,k-1}(\| {z}_k \|_\sigma)-{f}_k({v}_k)+{f}_{k-1}({v}_k)
\end{equation}
and which is based solely on a high accuracy speedometer for measuring $v_k(t)$ \footnote{For automotive applications, high accuracy speedometers are affordable and widely available. On the contrary, longitudinal accelerometers are notoriously unreliable and only used for purposes extraneous to navigation, such as the triggering of airbags in a collision event.} and on the measurements (\ref{z}), {\em locally} available to the $k$--th agent  (acquirable  for instance via onboard  LIDAR sensors\footnote{For automotive applications, commercially available affordable and high accuracy ``dot''  LIDARs  have latencies well under 1$\mu$s. Given the typical speeds of road vehicles, this implies that a numerical differentiation of the interspacing distance $z_k(t)$ in order to obtain the relative speed $z_k^v(t)$  is feasible via a high sampling frequency.}). Thus, the  control law  at the $k$--th agent reads:

\[
 u_k=u_{k-1}+u_k^\ell.
 \]

 In Figure~\ref{f2}, we denoted with $K_k$ the input--output operator from $z_k, z_k^v$ and $v_k$ respectively to $u_k^\ell$ of the $k$-th sub-controller from (\ref{FigureHelper}), namely

\begin{equation} \label{FigureHelperbis}
u_{k}^\ell = K_k \star \big ( z_k,\:  z_k^v, \: v_k\big).
\end{equation}
The resulted control architecture for any two consecutive agents ($k \geq 2$) can be pictured as in Figure~\ref{f2}.  For all practical purposes, the existence of a time delay on each of the feedforward links $u_k$, with $1 \leq k \leq (n - 1)$ must be taken into account.  For readability, these time delays are figuratively denoted  by $e^{-\theta s}$ in Figure~\ref{f2} (the Laplace transform of a delay of $\theta$ seconds), representative to the situation in which the delayed version $u_k(t-\theta)$ version of the $u_k(t)$ signal is received on board of the $(k+1)$ agent. In applications, these delays are caused by the physical limitations of the wireless communications system used for the implementation of the feedforward link , entailing a $\theta$ time delay at the receiver. For automotive applications the standard digital radio  communications systems (included in Figure~\ref{f2}) are DSRC\footnote{$\quad$ {\em IEEE 802.11p} - Dedicated Short Range Communications}.


\begin{rem}
Without the assumption of inter-agent communications delays, one might argue that information from the leader propagates instantaneously to all the agents in formation, via a relay mechanism (from each agent to its successor) and consequently the resulted distributed scheme doesn't employ local, but rather global information from the leader. It is known that precisely this type of time delays can  drastically alter the performance control architectures based on  such relay schemes \cite{rick}.
\end{rem}

\begin{rem}\label{DelSimpl}
 For the illustrative simplicity of the exposition, we look first at  the scenario in which there are no time--delays induce by the inter-agent (wireless) communication of information, such as the predecessor's control signal $u_{k-1}$. A  ``synchronization'' mechanism that can cope with the time-varying communications induced time--delays will be addressed in Section~\ref{delaycompensation}.
 \end{rem}


\subsection{A First Glance at the Closed--Loop Dynamics Decoupling}

 The control policy (\ref{oursbis}) entails a highly beneficial ``decoupling'' feature of the closed--loop dynamics at each agent, as  illustrated next. Firstly, note that by plugging (\ref{oursbis}) into (\ref{primo})  we obtain the following closed--loop error equations at the $k$--th agent:

\begin{equation} \label{secundo}
\dot{{z}}_k^v={{ f}_{k-1}({ v}_{k-1})- f_{k-1}}({ v}_k) - {\beta_k} {z}_k^v + \nabla_{{y}_k} V_{k,k-1}(\| z_k \|_\sigma).
\end{equation}

The following result will be instrumental in the sequel. Consider  the following Lyapunov candidate functions:

\begin{equation} \label{Lk}
\begin{split}
 L_k\big(z_k(t),z_k^v(t) \big)  \overset{def}{=} \frac{1}{2}  \Big(
V_{k,k-1}( \| {{ z}}_k(t) \|_\sigma) + \quad \quad \quad \quad \quad \quad \quad  \\
 \quad \quad \quad \quad \quad \quad \quad + {z}_k^v {}^\top(t) {z}_k^v(t) \Big),\ \text{with} \; 1 \leq k \leq n.
\end{split}
\end{equation}

\begin{lem}\label{instrumental-lemma}
The derivative of the Lyapunov candidate function $L_k(\cdot,\cdot)$ introduced in \eqref{Lk} along the trajectories of (\ref{primo}) and (\ref{oursbis}) is given by

\begin{equation} \label{Ldot22}
\begin{split}
\frac{d}{dt}L_k(z_k(t),z_k^v(t))&= {z}_k^v {}^\top(t)  \Big( f_{k-1}\big(v_{k-1}(t)\big)-f_{k-1}\big(v_k(t)\big) \Big) \\
&\quad \quad \quad \quad \quad \quad \quad \quad  - {\beta_k} {z}_k^v {}^\top(t) {z}_k^v(t) \: ,
\end{split}
\end{equation}
and does not depend on the choice of the  APFs $V_{k,k-1}(\cdot)$.
\end{lem}
\begin{IEEEproof}
Differentiating the APF $V_{k,k-1}(\cdot)$ at the $k$--th agent  with respect to time, yields

\begin{equation} \begin{split}
&\frac{d}{dt} V_{k,k-1}( \| { y}_{k-1} - { y}_{k} \|_\sigma) = {(\dot{{ y}}_{k-1}-\dot{{ y}}_{k})}^\top \times \\
& \big(\nabla_{{ y}_{k-1}} V_{k,k-1}( \| {y}_{k-1} -{y}_k \|_\sigma)-\nabla_{{ y}_{k}} V_{k,k-1}( \| {y}_{k-1} - {y}_k \|_\sigma) \big)
\end{split}
\end{equation}
and by employing the  anti--symmetrical property of APFs \cite[pp. 197]{SCL2012} :  $\nabla_{{ y}_k} V_{k,k-1}( \| {z}_k  \|_\sigma) =-\nabla_{{y}_{k-1}} V_{k,k-1}(\| {z}_k \|_\sigma)$  we get that

\begin{equation} \label{namol}
\frac{d}{dt}V_{k,k-1}( \| {z}_k \|_\sigma) = -2\: {\dot{{ z}}_k} {}^\top \: \nabla_{{ y}_{k}} V_{k,k-1}( \| {z}_k \|_\sigma).
\end{equation}

Therefore from (\ref{Lk}) it follows that

\begin{equation*}
\begin{split}
\frac{d}{dt}L_k\big(z_k(t),z_k^v(t) \big) &={z}_k^v {}^\top \dot{{z}}_k^v -{z}_k^v {}^\top \nabla_{{y}_k} V_{k,k-1}( \| {z}_k  \|_\sigma)  \\  &\hspace{-1.2cm}=  {z}_k^v {}^\top  \big(\dot{{z}}_k^v- \nabla_{{y}_k} V_{k,k-1}( \| {z}_k  \|_\sigma)  \big) \\
&\hspace{-1.2cm}\overset{ (\ref{secundo})}{=}  {z}_k^v {}^\top  \big(  { f}_{k-1}({ v}_{k-1})-{ f}_{k-1}({ v}_k)) - {\beta_k} {z}_k^v \big) \\
&\hspace{-1.2cm} =  {z}_k^v {}^\top  \big(  { f}_{k-1}({ v}_{k-1})-{ f}_{k-1}({ v}_k) \big) -  {\beta_k} {z}_k^v {}^\top {z}_k^v
\end{split}
\end{equation*}

\end{IEEEproof}


\section{Decoupling control design} \label{MR}


The following result is the main result of this Section, as it delineates a  ``decoupling'' property of the closed--loop dynamics, achieved by the  type (\ref{oursbis}) control policy along with: {\em (i)} closed-loop stability, {\em (ii)} velocity matching, {\em (iii)} collision avoidance and {\em (iv)} formation topology preservation. Specifically, assuming that the acceleration of the leader vehicle becomes zero after a finite period of time ({\em i.e.} $v_0(t)$ reaches a steady-state) then the following theorem holds:

\begin{theorem} \label{Nostra} If the functions $f_k(\cdot)$ with $0\leq k \leq n$ from (\ref{nlin}) satisfy the Lipshitz--like condition  \cite[Assumption~1]{SCL2012}

\begin{equation} \label{Lipsha}
\left|(v_2-v_1)^\top \big( f_k(v_2)-f_k(v_1)\big)\right| \leq \alpha_k \|v_2-v_1 \|_2^2, \quad \forall \:v_1,\: v_2
\end{equation}
then for any of the type \eqref{oursbis} control laws, such that $\beta_k>\alpha_{k-1}$,  the following hold: 

\noindent {\bf (A)} 
Given the Lyapunov function $L_k(\cdot, \cdot)$  from (\ref{Lk}), {\em local} to the $k$-th agent,
then for any real constant $c>0$ the sub--level sets $\Omega_c^k\overset{def}{=} \{ (z_k, z_k^v) | L_k (z_k, z_k^v) \leq c \}$ of $L_k(\cdot, \cdot)$ are compact and they represent forward invariant sets for the {\em local}  closed--loop dynamics (\ref{secundo}) of the $k$--th agent. \\
\noindent {\bf (B)} The control laws  \eqref{oursbis} guarantee  velocity matching in the steady-state {\em i.e.} $\displaystyle \lim_{t \rightarrow \infty} \|z_k^v(t)\| = 0$  and collision avoidance in the transient regime, {\em i.e.} there exists $\eta_c>0$ such that

$$\| z_k(t) \|_2>\eta_c, \; \forall t \geq0 .$$

\noindent {\bf (C)} The controller \eqref{ours} guarantees the formation's topology preservation in the steady-state, {\em i.e.}

$$\lim_{t \rightarrow \infty}\| z_k(t) \|_2=\delta_k$$
where $\delta_k$ is a pre-specified real, positive value.

\end{theorem}
\begin{IEEEproof}
 {\bf (A)}  We show that for any real $c>0$ the {\em local} sub--level sets $\Omega_c^k\overset{def}{=} \{ (z_k, z_k^v) |\ L_k (z_k, z_k^v) \leq c \}$ of $L_k(\cdot, \cdot)$ are compact. Note that $L_k(z_k, z_k^v)<c$ implies that $\|{z}_k^v \|< 2c$ and $V_{k,k-1}( \| z_k \|_\sigma)<2c$. Since $V_{k,k-1}(\cdot)$ is radially unbounded this implies that $ \| z_k \|_\sigma$ is bounded and consequently $ \| z_k \|_2$ is bounded. Therefore $\Omega_c^k\subset\R^{2 \text{dim} (y_k)}$ is a bounded set\footnote{Here $\text{dim}(y_k)$ denotes the dimension of the $y_k(t)$ vector valued function of time, which in general may be greater than one.}. Moreover due to the continuity of $ \| \cdot \|_\sigma$ and of $L_k(\cdot)$, one obtains that $\Omega_c^k$ is a closed set. Precisely $\Omega_c^k$ is the pre-image of a closed set through a continuous function. In the Banach space $\R^{2 \text{dim} (y_k)}$ it therefore holds that $\Omega_c^k$ is closed and bounded thus $\Omega_c^k$ is compact. Furthermore, Lemma \ref{instrumental-lemma} and the Lipschitz--like assumption (\ref{Lipsha}) on all $f_k(\cdot)$ implies that

\[
\frac{d}{dt}L_k\big(z_k(t),z_k^v(t)\big)\leq (\alpha_{k-1}-\beta_k){z}_k^v {}^\top(t) {z}_k^v(t),\ \forall k
\]
along the trajectories of (\ref{secundo}). Therefore it suffices to choose  the controller gain $\beta_k>\alpha_{k-1}$ in order to guarantee that along the trajectories of (\ref{secundo}) it holds that $\displaystyle\frac{d}{dt}L_k\big(z_k(t),z_k^v(t)\big)<0$  and also that $\displaystyle \Omega_c^k$ is a forward invariant set for the ``decoupled'' closed--loop system \eqref{secundo}, local to the $k$-th agent.\\

{\bf (B)} From the properties of the APF (Definition~\ref{apf}) it follows that  $V_{k,k-1}( \| z_k \|_\sigma)\rightarrow \infty$ as $ \| z_k \|_2\rightarrow 0$, {\em i.e.} $\forall c>0, \ \exists \eta_c>0$ such that

\begin{equation}\label{collision_avoidance}
V_{k,k-1}( \| z_k \|_\sigma)>c, \ \forall \ \| z_k \|_2<\eta_c.
\end{equation}
Let $c_k=\min_{r\ge 0}V_{k,k-1}(r)>0$. 
It follows from \eqref{collision_avoidance} that for any positive $c>c_k$ one has that

\begin{equation} \label{2collision_avoidance}
V_{k,k-1}( \| z_k \|_\sigma)\le c \mbox{ implies } \| z_k \|_2\geq\eta_c.
\end{equation}
Note that an increase of $c$ is correlated with a corresponding decrease of $\eta_c$.  Next, let us fix $c=2L_k(z_k(0),z_k^v(0))$. From point {\bf (A)} above it follows that for $\beta_k>\alpha_{k-1}$ it holds that $\Omega_c^k$ is a forward invariant set with respect to  \eqref{secundo} and consequently $L_k(z_k(t),z_k^v(t))\leq \frac{c}{2},\ \forall t\geq0$. This implies via (\ref{Lk}) that $V_{k,k-1}( \| z_k(t) \|_\sigma)<c,\ \forall t\geq0$ which in turn yields $c>c_k$ and so from \eqref{collision_avoidance} we conclude that

\[\| z_k \|_2>\eta_c,\ \forall t\geq0.\]
It is noteworthy that $\eta_c$ is implicitly defined by $c$ which in turn depends on the initial conditions $(z_k(0),z_k^v(0))$.

{\bf (C)}  Given $L_k(\cdot,\cdot)$ as introduced in (\ref{Lk}), it is claimed that the string formation's steady--state configuration is attained at the minimum of the following {\em formation-level Lyapunov function}, defined as

\begin{equation} \label{L}
L(z(t),z^v(t))\overset{def}{=} \frac{1}{2}  \sum_{k=1}^{n} L_k(z_k(t),z_k^v(t))
\end{equation}
where $z(t),z^v(t)$ are the aggregated vectors of the regulated measurements  for the entire formation, obtained by adequately stacking the local measurements $z_k(t),z_k^v(t)$ of the  agents:
\begin{equation} \label{zzz}
z(t) \overset{def}{=}\ba{cccc} z_1(t) & z_2(t) & \dots & z_{n}(t) \ea^\top
\end{equation}
\begin{equation} \label{zzzv}
z^v(t) \overset{def}{=}\ba{cccc} z_1^v(t) & z_2^v(t) & \dots & z_{n}^v(t) \ea^\top
\end{equation}
The minimum of (\ref{L}) therefore coincides (component--wise) with the minima of the Lyapunov functions (\ref{Lk})  {\em local} to the $k$-th agent.

In order to prove the claim, first note that the level sets of $L(\cdot,\cdot)$ given by $\Omega_c\overset{def}{=} \{ (z, z^v) | L (z, z^v) \leq c, \; \text{with} \: c>0 \}$ are compact, since $\Omega_c$ is a finite cartesian product of the $\Omega_c^k$  sets, whose compactness was proved at point {\bf (B)} above.  By a similar argument, it follows that $\Omega_c$ represents a forward invariant set for the closed--loop dynamics of the entire formation, along the trajectories of (\ref{secundo}), where $1 \leq k \leq n$.

Note that from the definition of (\ref{L}) and Lemma \ref{instrumental-lemma} that along the trajectories of (\ref{dacia}) and (\ref{oursbis}) one has that
\begin{equation} \label{SumLdot22}
\begin{split}
\frac{d}{dt}L\big( z(t), z^v(t)\big) =   \quad \quad \quad \quad    \quad \quad \quad \quad  \quad \quad \quad \quad  \quad \quad \quad \quad \\
  \quad \quad \quad  \sum_{k=1}^{N} {z}_k^v {}^\top  \Big( { f}_{k-1}({ v}_k) - { f}_{k-1}(({ v}_{k-1}) \Big) -  \sum_{k=1}^{N} {\beta_k} {z}_k^v {}^\top {z}_k^v \: .
\end{split}
\end{equation}
Therefore, by employing LaSalle's invariance principle we conclude that  the Lyapunov function $L_k(z_k(t), z_k^v(t))$ converges asymptotically to its minimum {\it i.e.} $\frac{d}{dt}L_k\big(z_k(t), z_k^v(t)\big)=0$, which is attained at velocity matching as $z_k^v=0$  (or equivalently when $v_{k-1}=v_k$) . Denote with $(\delta_k,0)$ the point at which such minimum is attained  and note that $\delta_k >0$ by the collision avoidance attribute from point {\bf (B)}. Consequently, in the steady-state $\|z_k^v(t)\|_2$ converges to $0$,  while $\|z_k(t)\|_2$ converges to $\delta_k$ and the formation's topology preservation
$$\lim_{t \rightarrow \infty}\| z_k(t) \|_2=\delta_k, \quad 1\leq k \leq n,$$
is guaranteed.
\end{IEEEproof}

\begin{rem}
The pre-specified positive constants $\delta_k$ (with $1 \leq k \leq n$) representing the desired steady-state  inter-agent distances, can be integrated in the Lyapunov functions $L_k(\cdot, \cdot)$  at the controller design stage,  as exemplified in Section~\ref{ANE}.
\end{rem}

\subsection{Communicating the Control Signals of the sub-Controllers}

When looking at the problem of agents moving in the string formation, the intuition behind the type \eqref{oursbis} control law is that it essentially includes an implementation of the common ``break lamp bulb'' regulated in all road traffic. This is the conceptual difference of the proposed  distributed architecture: instead of choosing to communicate the regulated measurements between sub-controllers (as done in  the cooperative output regulation setup such as \cite{3ZECISITREI, 3ZECISIPATRU, 3ZECISIUNU, 3ZECISIDOI}) the sub-controllers choose to transmit their control actions. This feature turns out to be essential in achieving the synchronization objective along with all the other standard performance requirements ({\em e.g.} collision avoidance, topology preservation) from displacement-based formation control \cite[Section~6]{formsurv}, as outlined by the following result.

\begin{prop}\label{NonInvariance}
 If the predecessor's control  action  $u_{k-1}(t)$ is not made available at the $k$-th agent in the controller equation (\ref{oursbis}) then velocity matching cannot be attained.
 \end{prop}
 \begin{IEEEproof}
 We prove the result by contradiction. Let us assume that velocity matching is achieved. Consequently, in the steady-state $\|{z}_k^v(t)\|_2$ converges to zero and $\|{z}_k(t)\|_\sigma$ converges to some constant $\eta\geq\delta_k$. Since $V_{k,k-1}(\cdot)$ is a class $C^1$ function, all its partial derivatives are continuous and so $\nabla_{{y}_k} V_{k,k-1}(\| z_k \|_\sigma)$ converges to a vector

 $$\nabla_\eta \overset{def}{=} \nabla_{{y}_k} V_{k,k-1}(\| z_k \|_\sigma)\big|_{\|{z}_k(t)\|_\sigma=\eta}$$
 Whenever $\|{z}_k^v(t)\|_2$ converges to zero, it follows that

\begin{equation}\label{eq:n1}
\forall c>0,\ \exists \ t_0>0 \mbox{ such that } \|{z}_k^v(t)\|_2\le c, \ \forall t\ge t_0.
 \end{equation}
 In order to show that \eqref{eq:n1} is violated, let $c>0$ be fixed and $t_0=\max\{t>0\mid \|{z}_k^v(t)\|_2= c\}$. Therefore, we should have $\|{z}_k^v(t)\|_2< c,\ \forall t\ge t_0$.\\[2mm]
Since $\nabla_{{y}_k} V_{k,k-1}(\| z_k \|_\sigma)$ converges to $\nabla_\eta$ one has that $\exists \epsilon>0$ such that

\[
\big|{z}_k^v {}^\top(t)\big(\nabla_{{y}_k} V_{k,k-1}(\| z_k \|_\sigma)-\nabla_\eta\big)\big|<\epsilon c
\]
Employing \eqref{Lipsha} (while noting  that $u_{k-1}$ is implicitly present in the closed-loop equations \eqref{secundo}), one obtains:

\begin{equation*} \label{LdotPerturbed}
\begin{split}
&\frac{d}{dt} \|{z}_k^v(t)\|_2= {z}_k^v {}^\top(t)  \Big( f_{k-1}\big(v_{k-1}(t)\big)-f_{k-1}\big(v_k(t)\big) \Big) \\
&- {\beta_k} {z}_k^v {}^\top(t) {z}_k^v(t)+  {z}_k^v {}^\top(t)\nabla_{{y}_k} V_{k,k-1}(\| z_k \|_\sigma)\\
&+ {z}_k^v {}^\top(t)u_{k-1}(t)\: \\
&\geq -(\alpha_{k-1}+\beta_k) {z}_k^v {}^\top(t) {z}_k^v(t)-\epsilon c+ {z}_k^v {}^\top(t)\big(\nabla_\eta+u_{k-1}(t)\big)\:.
\end{split}
\end{equation*}
Therefore, while

\begin{equation}\label{destabilize u}
u_{k-1}(t)=(\beta_k+\alpha_{k-1}+2\epsilon){z}_k^v(t)-\nabla_\eta
\end{equation}
it follows that $\displaystyle\frac{d}{dt} \|{z}_k^v(t_0)\|_2>0$ which in turn implies that \eqref{eq:n1} doesn't hold. It is important to note here that no increase in the control  effort by choosing a larger gain $\beta_k$ can help in achieving velocity matching, since for arbitrarily  large but finite values of $\beta_k$, under the assumption that $\|{z}_k^v(t)\|_2$ converges to zero, one gets that the term $(\beta_k+\alpha_{k-1}+2\epsilon){z}_k^v(t)$ becomes arbitrarily small. Nevertheless, once this term becomes small, there always exists a  bounded contro input $u_{k-1}$ of the form \eqref{destabilize u} that precludes $\|{z}_k^v(t)\|_2$ from converging to zero.
\end{IEEEproof}

 \subsection{Some Considerations and Future Work} \label{CDCA}
  For illustrative simplicity, we provide below an informal outlook of our approach leading to the proposed distributed controller architecture for the string network. We look at how the aggregated {\em variables for the entire formation} relate to the  agents' {\em individual variables} (the states - including the absolute coordinates $y_k$ and the speeds $v_k$, the sub-controller's outputs  $u_k$  and  the measurements $z_k, z_k^v$ local to the $k$-th agent) . Let us define:

\begin{equation} \label{v}
v(t) \overset{def}{=}\ba{cccc} v_0(t) & v_1(t) & \dots & v_n(t) \ea^\top
\end{equation}
\begin{equation} \label{u}
u(t) \overset{def}{=}\ba{cccc} u_0(t) & u_1(t) & \dots & u_n(t) \ea^\top
\end{equation}
\begin{equation} \label{f}
f\big(v(t) \big) \overset{def}{=}\ba{cccc} f_0(u_0) & f_1(u_1) & \dots & f_n(u_n) \ea^\top
\end{equation}
while the vectors of the aggregated measurements $(z,z^v)$ have been introduced in (\ref{zzz}), (\ref{zzzv}). The aggregated equations (\ref{nlin}) for the entire formation become

\begin{equation} \label{nlinbis}
\dfrac{d}{dt}{{ v}}(t)=f({ v}(t)) + { u}(t)
\end{equation}

The objective of the control scheme is related to the regulated measurements (in our case $z^v$), which (in general) are functions of the agent's states (in our case $v$). Let us describe this dependence by

\begin{equation} \label{J}
  z^v= \mathcal{J}(v)
\end{equation}
In many situations of practical interest the $\mathcal{J}(\cdot)$ operator may be linear.  For the situation studied here, it follows from (\ref{z}) that $\mathcal{J}(\cdot)$ is an  $n \times (n+1)$ real matrix having entries equal to ``$1$'' on its diagonal and ``$-1$'' on its supra-diagonal. In displacement-based formation control the type (\ref{J}) definition of the regulated variables ({\em e.g.} relative distances or relative speeds) encapsulates the topology of the multi-agent formation. In related formulations ({\em e.g.} the optimal control formulation for LTI agents from \cite{TAC2016}) the $\mathcal{J}(\cdot)$ operator may also enclose norm based costs.

Let us assume next that $\mathcal{J}(\cdot)$  commutes with the differentiating operator  $\displaystyle d/dt$ and apply $\mathcal{J}(\cdot)$ to both sides of (\ref{nlinbis}), while taking into account its linearity and the definition of the regulated variables (\ref{J}) in order to get
\begin{equation} \label{tertz}
\dfrac{d}{dt}\mathcal{J}(v(t))=\mathcal{J} \Big( f({ v}(t)) \Big) + \mathcal{J} \Big( { u}(t) \Big)
\end{equation}

The merit of the type (\ref{tertz}) formulation is that it is expressed directly in terms of the regulated measurements $z^v= \mathcal{J}(v)$ (on the left-hand side) and not of the original states $v$ from (\ref{nlinbis}), while the dynamics of  ``the plant'' have changed from $f(\cdot)$ to $\mathcal{J}(f(\cdot))$. The morale is that the distributed controller design may now be performed on (\ref{tertz}) in order to obtain the control laws $\mathcal{J}(u)$. This is exactly the approach taken for the string formation, where for the regulated measurements $z_k^v=v_{k-1}-v_k$ we have essentially performed a decentralized controller synthesis (\ref{ours}) for the control signals $(u_{k-1}-u_k)$. The relation from (\ref{tertz}) also suggests that the sub-controller communications topology should ``borrow'' the formation's topology. This entails  the ``dynamic decoupling'' attained by the string formation,  the existence of  ``local'' forward invariant sets and Lyapunov functions and the possibility of performing independent regulation and sub-controller design at each agent.

Finding suitable distributed controllers for more general types of  $\mathcal{J}(\cdot)$ operators is the objective of future investigations. Assuming that the distributed controller design for (\ref{tertz}) successfully yields the control laws $\mathcal{J}(u)$, there still remains the problem of finding a {\em causal implementation} of such a controller\footnote{This aspect will appear even more more clearly in those formulations of these type of problems involving dynamical systems with discrete-time.} in terms of the  $u$ signals of the original formulation (\ref{nlinbis}). This may require certain invertibility assumptions on the $\mathcal{J}(\cdot)$ mapping. Furthermore, looking at any practical implementation of the proposed law from (\ref{oursbis}), the local control $u_k(t)$ should only depend on the delayed version of $u_{k-1}(t)$ (or on its history). These important issues will be discussed in Section~\ref{delaycompensation} below.

\subsection{Platooning Control}
Platooning control has been a longstanding problem in control engineering, encompassing a vast literature. For a series of recent, interesting results also providing a good  outline of existing literature we refer to \cite{GaborI, GaborII}. To the best of the authors's knowledge, the current results are the only ones  guaranteeing collision avoidance and topology preservation for heterogenous, nonlinear dynamical agents in the presence of communication induced time-delays, as outlinen in Section~\ref{delaycompensation} below. In this context, the current paper can be looked at as an extension to nonlinear control of the novel ideas for $\mathcal{H}_2 / \mathcal{H}_\infty$ control of LTI agents in \cite{TAC2016}.

\section{A  Practical Time-Delays Compensation Mechanism}\label{delaycompensation}

The difficulties caused  for networked systems by the communications induced delays and time jittering have been a topic of intensive study for decades. In formation control practical applications it has been argued in \cite{rick} that the (relatively low latency) time delays induced by the wireless communications of the control signals $u_k$ from one agent to its successor in the string (even if assumed time-invariant and homogeneous\footnote {For digital radio wireless systems such as WiFi, Bluetooth or Zigbee, the corresponding time-delays have low latencies but they are time-varying, taking values around a nominal delay of about 20 ms.}) irremediably alter the performance of the control scheme. What happens is that the delays propagate through the closed-loops  towards the back of the platoon where they ``accumulate'' in a manner depending on the number of vehicles in formation \cite{rick}.

For the case of LTI  dynamical agents, the very recent results from  \cite{TAC2016} provide a functional solution for compensating the effect of the communications delays  by introducing (with high accuracy) supplemental delays at key points in the closed loop. However, the aforementioned method \cite[Section~VI]{TAC2016} of essentially ``moving'' the  {\em synchronization delay} through the loop and ultimately incorporating it in the model of the ``plant'' cannot be directly adapted to nonlinear dynamical systems.  In this section we show how an adaptation of  the distributed controller of Section~\ref{MR}  is able to compensate the communications delays, while essentially preserving all the performance features from the delay free case. The approach taken here is based on: {\em (i)} the tailored use of the so-called {\em time-headways} without sacrificing the tightness of the formation  and {\em (ii)} changing the definition of regulated measurements to a meaningful approximation. The main challenge here is for the re-defined regulated measurements to remain measurable (on board of each agent) in a distributed manner.

\subsection{Adapting Time-Headways for Delays Compensation}

Classical results in platooning control \cite{TH1, TH2, TH3} proved that a considerable improvement of performance can be obtained by adequately modifying the regulated interspacing distance  (for each vehicle $k$)  $z_k=y_{k-1}-y_k$ such as to include a factor $- h \dot{y}_k(t)$  proportional with the speed of the current vehicle. The resulted interspacing policies (dubbed time-headways)   become $z_k=y_{k-1}(t)-y_k(t) - h\dot{y}_k(t)$ (where $h>0$, the so called time-headway, is a real, positive constant) and provide a spacing  in time rather than distance  (between two consecutive vehicles). Up until the recent distributed scheme introduced in \cite{TAC2016} -  in the case of LTI dynamical agents - good attenuation at all frequencies could only be achieved via the use of time headways policies \cite{TAC2010}. The generally adopted value for highway platooning (which became the standard at some point)   is $h =1$ second. The main drawback of such large time headways is that  they destroy the tightness of the formation,  drastically reducing the highway traffic throughput or any potential fuel savings achievable by the air drag reduction.

We introduce next a novel method for delays compensation that combines the synchronized-clocks mechanism from \cite[Section~VI]{TAC2016} with an adequate adaptation of the time headways. Firstly, let us revamp as follows the definitions  (\ref{z}) of the interspacing distances $z_k$ and of the regulated relative speeds $z_k^v$ respectively at the $k$-th vehicle:

\begin{equation} \label{znew}
\begin{split}
&{{ \tilde z}}_k(t)\overset{def}{=}{y}_{k-1}(t-\theta)-{y}_k(t-\theta) - \theta \dot{y}_k(t-\theta), \\
&{{ \tilde z}}_k^v(t)\overset{def}{=}{ v}_{k-1}(t-\theta)-{ v}_k(t-\theta) - \theta \dot{v}_k(t-\theta),
 1 \leq k \leq n, \\
\end{split}
\end{equation}

\noindent where the positive constant $\theta > 0$, {\em i.e.} \underline{the time-headway, will}  \underline{be taken to be equal with the communications delay} and will be considered (without any loss of generality) to be the same for all vehicles in formation\footnote{See also Remark~\ref{wifi} below}. It can be seen that the signals defined in (\ref{znew}) are  merely $\theta$ delayed version of (\ref{z}) with a $\theta$ time-headway.

If at the current moment in time $t$, we would choose to regulate instead the measurements taken at moment $(t-\theta)$ according to (\ref{znew}), that would be a limitation imposed by the communications delay (which are relatively very small, though) and it would entail some loss in performance which was to be expected. The inclusion of the $\theta$ time headway results in a slightly more conservative policy, since it induces slightly larger  interspacing distances as the speed increases. The same conservative effect (of the   $\theta$ time headway)  occurs with respect to the regulated  relative speeds ${{ \tilde z}}_k^v(t)$ during the transient regime when the acceleration $\dot{v}_k$ is sizable.

\begin{rem} \label{wifi}
For all practical applications related to platooning, the value of $\theta$ will be taken to be equal to a worst case scenario value of the latency of the wireless communication systems, which is  about $2\times 10^{-2}$ seconds for digital radio systems such as DSRC, WiFi, Bluetooth or Zigbee. Furthermore, the ``synchronized clocks''  mechanism introduced in \cite[Section~VI]{TAC2016} used in conjunction with time stamping protocols (at the transmission of the predecessor's control signal $u_{k-1}$) is able to emulate and implement time invariant and heterogeneous communications time delays throughout the entire formation, by introducing with high accuracy supplemental delays in the closed-loop.
\end{rem}

\begin{rem}\label{equinox}
The effect of the time-headway on the behavior of the formation is directly proportional with the numerical value of $\theta$, which is very small in practice\footnote{See Remark~\ref{wifi} above.}. In fact, for formation control practical applications the effect is almost negligible given the order of magnitude of $\theta$ compared to the time constants of the dynamics of road or aerial agents.
\end{rem}

 \begin{figure*}[!ht]
\hspace{-3mm}
\centering
\includegraphics[scale=0.85]{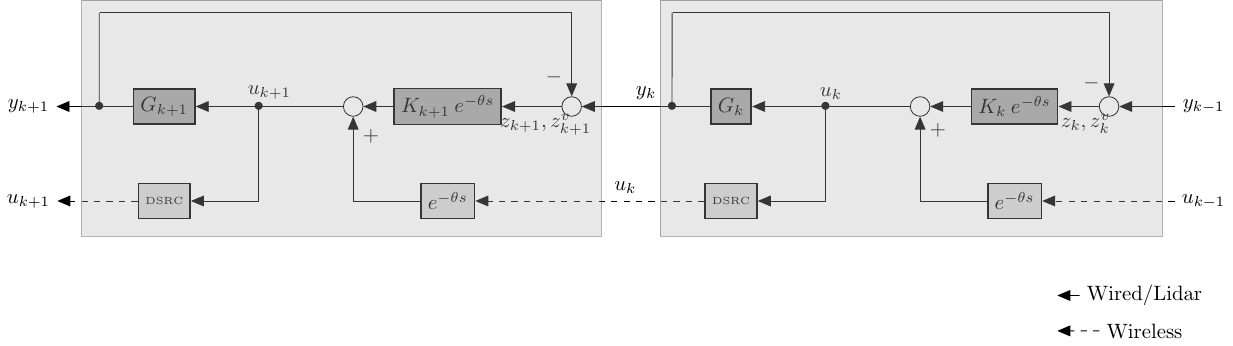}
\caption{Distributed Controller Implementation with Delays Compensation.} 
\label{f3}
\end{figure*}

\subsection{Changing the Regulated Measurements}\label{wind}

In the next proposition we re-define the regulated measurements and make the express remark that the definition included below will be enforced onward, throughout the contents of the current Section~\ref{delaycompensation}.

\begin{prop}\label{taitai}
We assume the following initial conditions
\[
y_k(t)= -\sum_{j=0}^k \ell_j, \quad v_k(t)=0,\ \forall t\in(-\theta,0],
\]
and  further define the regulated measurements to be:
\begin{equation} \label{z-delayed}
\begin{split}
\Aboxed{&{{ z}}_k(t)\overset{def}{=}{y}_{k-1}(t-\theta)-{y}_k(t)} \\
\Aboxed{&{{ z}}_k^v(t)\overset{def}{=}{ v}_{k-1}(t-\theta)-{ v}_k(t) \quad \text{for} \quad 1 \leq k \leq n.}
\end{split}
\end{equation}
It follows that ${{ z}}_k(t), {{ z}}_k^v(t)$ from (\ref{z-delayed}) above represent an  $\mathcal{O} (\theta^2)$ approximation of ${{ \tilde z}}_k(t)$ and ${{ \tilde z}}_k^v(t)$ respectively, from (\ref{znew}).
\end{prop}
\begin{IEEEproof}
It follows by the very definition of (\ref{z-delayed}) and (\ref{znew}) respectively and the following  Taylor series expansion:
\begin{equation} \label{approx}
\begin{split}
&{y}_k(t)={y}_k(t-\theta) + \theta \dot{y}_k(t-\theta) + \mathcal{O} (\theta^2),\\
&{v}_k(t)={v}_k(t-\theta) + \theta \dot{v}_k(t-\theta) + \mathcal{O} (\theta^2).
\end{split}
\end{equation}
\end{IEEEproof}

\begin{rem}\label{measure}
One essential practical issue is to establish if the new regulated signals (\ref{z-delayed}) remain measurable on board of the $k$-th agent (in a distributed manner). Writing in (\ref{z-delayed})   the Taylor series expansion with an integral rest for $y_k(t)$ and $v_k(t)$ respectively, we obtain the following equivalent expressions:
\begin{equation} \label{z-delayed-equiv}
\begin{split}
&{{ z}}_k(t)\overset{}{=}{y}_{k-1}(t-\theta)-{y}_k(t-\theta)-\int_{t-\theta}^t\dot{y}_k(\tau) d \tau, \\
&{{ z}}_k^v(t)\overset{}{=}{ v}_{k-1}(t-\theta)-{ v}_k(t-\theta)-\int_{t-\theta}^t\dot{v}_k(\tau) d \tau
\end{split}
\end{equation}
or equivalently
\begin{equation} \label{z-delayed-equiv-bis}
\begin{split}
&{{ z}}_k(t)\overset{}{=}{z}_{k}(t-\theta)-\int_{t-\theta}^t v_k(\tau) d \tau, \\
&{{ z}}_k^v(t)\overset{}{=}{ z_k^v}(t-\theta)-v_k(t) + v_k(t-\theta)
\end{split}
\end{equation}
and so it becomes apparent that the signals introduced in (\ref{z-delayed}) can be measured   on board the $k$-th vehicle via (\ref{z-delayed-equiv-bis}), using only onboard ranging sensors\footnote{Preferably very low latency LIDAR sensors, which are already affordable and widely available commercially} and high accuracy longitudinal speedometers in conjunction with a mere integrator. Specifically, the first term in (\ref{z-delayed-equiv-bis}) consists of the $\theta$-delayed measurement of the interspacing distance minus the integration of the absolute speed (measurable on board) over a $\theta$-length interval. The second term in (\ref{z-delayed-equiv-bis}) consists of the $\theta$-delayed measurement of the relative speed\footnote{The relative speed with respect to the preceding vehicle, which is also measurable onboard, see also footnote number 9 on page four.}  minus the $(v_k(t)-v_k(t-\theta))$ term,  comprised of absolute speeds measurable onboard the $k$ agent. Finally, the entire history on the interval $[ (t-\theta), \; t]$  of the ranging sensors (\ref{z-delayed-equiv-bis})  must be stored in a memory buffer, in order to be used by the distributed controller we will introduce next.
\end{rem}

The following remark represents the conclusion of the current Subsection~\ref{wind}:
\begin{rem}
Given the values of $\theta$ that appear in practice (see Remark~\ref{wifi}) and given the worst case scenario of breaking decelerations $| \dot{v}_k(t) |$ that could occur during highway traffic, it follows via Proposition~\ref{taitai} that the signals from (\ref{z-delayed}) are such an accurate approximation of  (\ref{z-delayed-equiv-bis}),  that the order of the approximation falls way below the tolerated measurement errors of the most performant ranging sensors. That is to say that choosing between two controllers that  regulate either the (\ref{znew}) signals or the (\ref{z-delayed-equiv-bis}) signals respectively, has considerably less influence on the resulted scheme than the measurement noise of an highly accurate LIDAR. Consequently, in the  scheme proposed next, we  choose to regulate  (\ref{z-delayed-equiv-bis}).
\end{rem}

\subsection{A Controller to Cope with  Time Delays}
Considering the definition of $z_k$ and $z_k^v$ as in (\ref{z-delayed}), we will prove that the distributed control policies given next  are able to entirely compensate  the communication induced delays:

\begin{equation} \label{ours-del}
\left\{\begin{split}
{ u}_k(t) = & { u}_{k-1}(t-\theta) + {\beta_k} {{ z}}_k^v(t)   + \nabla_{{y}_k} V_{k,k-1}(\| {{ z}}_k(t)\|_\sigma)\\
-&{f}_k({v}_k)+{f}_{k-1}({v}_k)\\
{ u}_k(t) = &0,\quad \forall t\in(-\theta,0]
\end{split}\right.
\end{equation}

\begin{rem} \label{implement} Note that for the real time, practical implementation of type (\ref{ours-del}) control policies onboard the $k$-th vehicle, two pieces of information are needed: {\em (i)} the  command signal  ${ u}_{k-1}(t-\theta)$ of the predecessor,  received on board via (DRSC) wireless communications with a $\theta$-delay  and {\em (ii)}  the (\ref{z-delayed-equiv-bis})  sensor measurements $z_k, z_k^v$, which are measurable on board the $k$-th agent, as per the considerations from Remark~\ref{measure}. Note that in order for the scheme to be effective,  the $\theta$ communications delay from ${ u}_{k-1}(t-\theta)$ must be replicated with high accuracy  in the in $z_k(t), z_k^v(t)$ measurements from  (\ref{z-delayed-equiv-bis}). This resembles the (GPS time-base) synchronization mechanism in \cite[Section~VI]{TAC2016}) for the LTI case and it can be implemented using time-stamping protocols of the involved signals ${ u}_{k-1}$ and $z_k(t), z_k^v(t)$. In Figure~\ref{f3}, such $\theta$ ``synchronization'' delays to be imposed on the $z_k(t), z_k^v(t)$ signals from  (\ref{z-delayed-equiv-bis}) have been figuratively incorporated in the controller, via the $e^{-\theta s}$ term. The aforementioned synchronization will   ensure \underline{time invariant, point-wise delays of value exactly $\theta$}, \underline{homogeneously  throughout the entire formation} as per the considerations from Remark~\ref{wifi}.
\end{rem}

We employ the Lyapunov function from \eqref{Lk} keeping in mind that the definitions of ${{ z}}_k(t),{{ z}}_k^v(t)$ are in accordance to (\ref{z-delayed}). Specifically, assuming that the acceleration of the leader vehicle becomes zero after a finite period of time ({\em i.e.} $v_0(t)$ reaches a steady-state), the time delays adaptation for the main result of Section \ref{MR} reads:
\begin{theorem} \label{Nostra-del} If the function $f(\cdot)$ from (\ref{nlin}) satisfies the global Lipshitz--like condition \eqref{Lipsha} then for any of the type \eqref{ours-del} control laws, such that $\beta_k>\alpha_{k-1}$  the following hold: \\
\noindent {\bf (A)} The derivative of the Lyapunov candidate function $L_k(\cdot,\cdot)$ from \eqref{Lk}, local to the $k$-th agent, along the trajectories of (\ref{primo}) and (\ref{ours-del}) is given by

\begin{equation} \label{Ldot22new}
\begin{split}
&\frac{d}{dt}L_k(z_k(t),z_k^v(t))= - {\beta_k} {z}_k^v {}^\top(t) {z}_k^v(t)+\\
&\qquad {z}_k^v {}^\top(t)  \Big( f_{k-1}\big(v_{k-1}(t-\theta)\big)-f_{k-1}\big(v_k(t)\big) \Big) \: ,
\end{split}
\end{equation}
and does not depend on the choice of the  APFs $V_{k,k-1}(\cdot)$. \\
\noindent {\bf (B)}
Given the Lyapunov function $L_k(\cdot,\cdot)$ from (\ref{Lk}),
{\em local} to the $k$-th agent,  the sub--level sets $\Omega_c^k\overset{def}{=} \{ (z_k, z_k^v) | L_k \leq c, \; \text{with} \: c>0 \}$ of $L_k$ are compact and they represent forward invariant sets for the {\em local}  closed--loop dynamics of the $k$--th vehicle. \\
\noindent {\bf (C)} The control laws  \eqref{ours-del} guarantee  velocity matching in the steady-state {\em i.e.} $\displaystyle \lim_{t \rightarrow \infty} \|z_k^v(t)\| = 0$  and collision avoidance in the transient regime, {\em i.e.} there exists $\eta_c>0$ such that

$$\| z_k(t) \|_2>\eta_c, \; \forall t \geq0 .$$
The controller \eqref{ours-del} also guarantees the formation's topology preservation in the steady-state, {\em i.e.}

$$\lim_{t \rightarrow \infty}\| z_k(t) \|_2=\delta_k$$
where $\delta_k$ is a pre-specified real, positive value.

\end{theorem}
\begin{IEEEproof} {\bf (A)}
With the controller (\ref{ours-del}) at hand we obtain the following closed--loop  equations at the $k$--th agent:
\begin{equation} \label{secundo-del}
\begin{split}
\dot{{z}}_k^v(t)&={{ f}_{k-1}({ v}_{k-1}(t-\theta))- f_{k-1}}({ v}_k(t))\\ &- {\beta_k} {z}_k^v(t) + \nabla_{{y}_k} V_{k,k-1}(\| z_k (t)\|_\sigma).
\end{split}
\end{equation}
Let us notice that we deal with one fixed delay. Therefore, its derivative is 0 and consequently no complexity is added to the computations with respect to the proof of  Lemma \ref{instrumental-lemma}. One straightforwardly gets that

\[
\frac{d}{dt}V_{k,k-1}( \| {z}_k \|_\sigma) =- 2\: {\dot{{ z}}_k} {}^\top \: \nabla_{{ y}_{k}} V_{k,k-1}( \| {z}_k \|_\sigma)
\]
with ${z}_k$ as defined in \eqref{z-delayed}. Since $\dot{{ z}}_k={z}_k^v$ it follows that

\begin{equation*}
\begin{split}
&\frac{d}{dt}L_k(z_k(t),z_k^v(t))= - {\beta_k} {z}_k^v {}^\top(t) {z}_k^v(t)+\\ &\qquad {z}_k^v {}^\top(t)  \Big( f_{k-1}\big(v_{k-1}(t-\theta)\big)-f_{k-1}\big(v_k(t)\big) \Big) \: ,
\end{split}
\end{equation*}
{\bf (B)} Notice that since $f_{k-1}(\cdot)$ satisfies \eqref{Lipsha} it follows that
 \[\begin{split}
 {z}_k^v {}^\top(t)  \Big( f_{k-1}\big(v_{k-1}(t-\theta)\big)-f_{k-1}\big(v_k(t)\big) \Big)\leq\\
 \alpha_{k-1} {z}_k^v {}^\top(t) {z}_k^v(t)
\end{split} \]
Consequently, along the trajectories of the closed-loop system \eqref{secundo-del} one has

\[
\frac{d}{dt}L_k(z_k(t),z_k^v(t))\leq (\alpha_{k-1}-\beta_k){z}_k^v {}^\top(t) {z}_k^v(t).
\]
Choosing $\beta_k>\alpha$ we guarantee that $\frac{d}{dt}L_k(z_k(t),z_k^v(t))<0$ along the trajectory of \eqref{secundo-del}. Moreover, it follows along the lines of the proof of Theorem \ref{Nostra} that $\Omega_c^k$ is compact and forward invariant.\\
{\bf (C)} Along the lines of the proofs of points {\bf (B)} and {\bf (C)} of Theorem \ref{Nostra} we can show that there exists $\eta>0$ such that

$$\| {z}_k\|_2>\eta,\quad\forall t\geq0$$
and the steady state value is given by

$$\lim_{t \rightarrow \infty}\| {z}_k(t) \|_2=\delta_k.$$
In order to obtain the desired results we have just to notice that for all $k$ the function ${y}_{k-1}(t)$ is non-decreasing in time. Consequently,

\[
\| {y}_k(t) - {y}_{k-1}(t) \|_2\geq \| {y}_k(t) - {y}_{k-1}(t-\theta) \|_2>\eta, \ \forall t>0
\]
and
\[
\lim_{t \rightarrow \infty}\| {y}_k(t) - {y}_{k-1}(t) \|_2\geq\lim_{t \rightarrow \infty}\| {z}_k(t) \|_2=\delta_k
\]

\end{IEEEproof}
\begin{rem}\label{fine} The scheme proposed above is able to regulate ${ v}_{k-1}(t-\theta)-{ v}_k(t)$ in the presence of communications delays. Therefore, as far as the leader's velocity profile is  slowly varying relatively to the order of magnitude of the communications delays, the scheme does regulate an accurate approximation of ${ v}_{k-1}(t)-{ v}_k(t)$. Nevertheless,  oscillations of the leader's velocity at a frequency that is of the same order of magnitude with  $1/\theta$, cannot be  efficiently compensated and the accordion effect will appear. These assumptions are very well satisfied in the platooning setting, but they may not be valid for other applications. The conclusion is in line with the well known fact that for the validity of the control scheme it is always necessary that the time delays that propagate through the controller are smaller than those propagating through the given plant.
\end{rem}

\section{A Numerical Example and Further Considerations} \label{ANE}

In this section we illustrate the distributed controllers introduced above for dynamical agents (\ref{dacia}) where the function $f_k(\cdot)$ has a quadratic form $f_k(v)= -\gamma_k g -\ell_k v^2$, in accordance with the dynamical model of road vehicles from \cite[(1)/pp. 1]{Gabor}. Here, $g=9.81\: m/s^2$ is the gravitational acceleration, $\gamma_k$ is the tyre rolling resistance coefficient and $\ell_k$ the air drag constant of vehicle $k$. The dynamics (\ref{dacia}) are

\begin{subequations}\label{daciaLucian}
\begin{equation} \label{daciaLucianA}
\dot{y}_k= v_k
\end{equation}
\begin{equation} \label{daciaLucianB}
\dot{v}_k= -\gamma_k g -\ell_k v_k^2+\frac{\eta}{R} w_k
\end{equation}
\end{subequations}
with $\eta$ the gear ratio and $R$ the wheel radius. The command signal $\omega_k$ is the engine's torque, and its linear transformation $\frac{\eta}{R} \omega_k$ corresponds to the input $u_k$ in (\ref{dacia}). We take the same values for $\eta=1.8$ and $R=0.5$ across all the vehicles, to preserve the same meaning of the input $u$. Note that for bounded velocities $\vert v_1 \vert, \vert v_2 \vert\leq v_{\max}$, $f_k(\cdot)$ in (\ref{daciaLucian}) satisfies the Lipschitz--like condition \cite[Assumption~1]{SCL2012} $(v_2-v_1)^\top \big( f_k(v_2)-f_k(v_1)\big) \leq \alpha_k \vert v_2-v_1 \vert^2$, with $\alpha_k = 2 \ell_k v_{\max}$. We take $v_{\max} = 60$\,m/s (i.e.\ 216\,km/h) in our experiments.

The control law is designed using APFs (Definition~\ref{apf}) of the following form \cite[Fig.~1/ pp.197]{SCL2012}

\begin{equation}\label{apfLucian}
V_{k,k-1}(\|z_k  \|_\sigma) = \ln (\|z_k \|_\sigma)^2+\dfrac{100}{\| z_k \|_\sigma^2}
\end{equation}
and a gain $\beta = 100$ is chosen, which will be greater than $\alpha_k$ for all vehicle parameters below. 
The reference signal for the entire formation is imposed by the control of the leader vehicle, namely $w_0(t)$, which consists of three smoothed rectangular pulses between the levels 15 and 30\,Nm. There are six vehicles in total including the leader, and they start at relatively small separations of about $2$\,m, with an initial velocity of $10$\,m/s.

For our baseline experiment, we take homogeneous vehicle dynamics with $\gamma_k=0.011$ and $\ell_k=0.463$\,kg/m for all vehicles $k$, and no delay. Note that in this case the term $-f_k(v_k) + f_{k-1}(v_k)$ from (\ref{oursbis}) disappears. Figure \ref{fig:baseline}, top shows the states and controls of the vehicles using absolute values to give an idea of their true trajectories; in Figure \ref{fig:baseline}, bottom the corresponding interspacing distances $z_k$, velocity differences $z_k^v$, and relative controls $u_k^\ell$ are reported. All subsequent figures will use such relative values. The baseline results show how the controllers initially prioritize increasing interspacing distances, after which the velocities are brought together. Note that on the timeline of this experiment the interspacings have not yet converged; by allowing the experiment to run longer we have confirmed that the steady-state interspacings are 10.95\,m, equal to the minimum of the potential function chosen (the trajectories are not shown here since they are not much more informative than \figref{fig:baseline}).
\begin{figure}[!htb]
\centering
\includegraphics[width=\columnwidth]{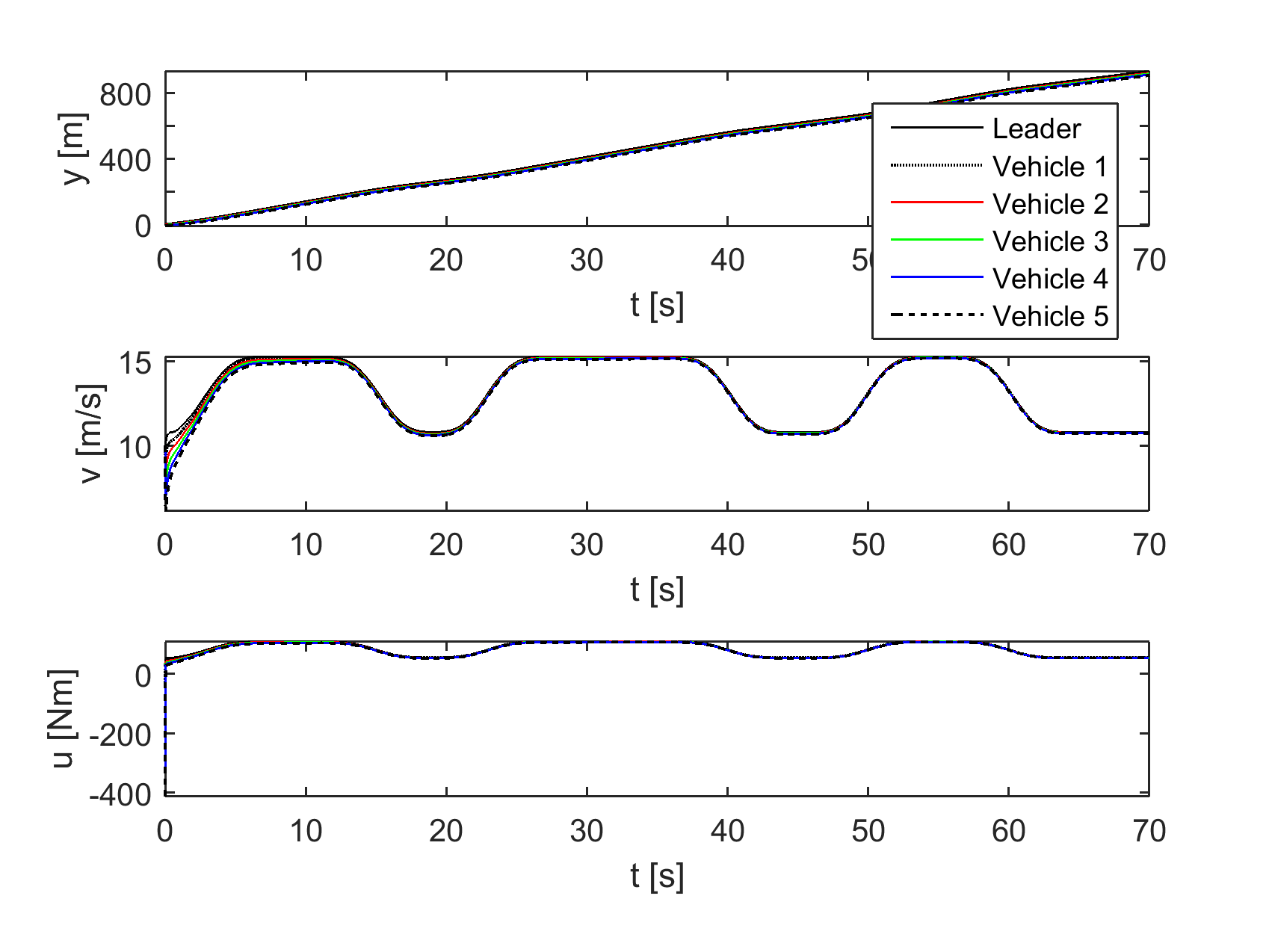}\\
\includegraphics[width=\columnwidth]{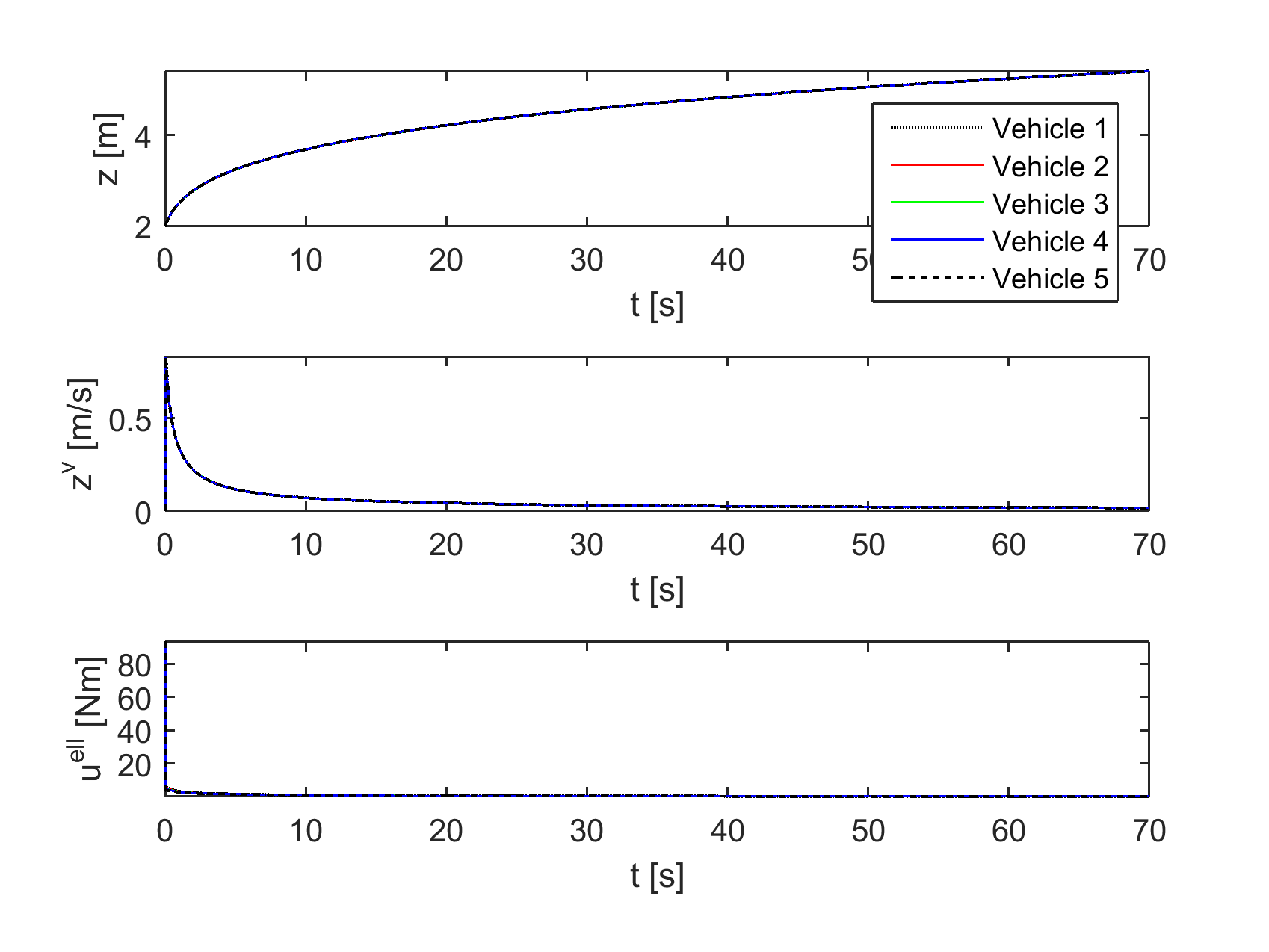}
\caption{Trajectories of the vehicles in the baseline case (homogeneous, no delay). Top: absolute values, bottom: relative inter-vehicle values.}
\label{fig:baseline}
\end{figure}

To verify Proposition \ref{NonInvariance}, the previous vehicle's input $u_{k-1}$ is removed from the control law of vehicle $k$. \figref{fig:nouprev} shows that, indeed, velocity agreement cannot be achieved, leading to divergence of the positions.
\begin{figure}[!htb]
\centering
\includegraphics[width=\columnwidth]{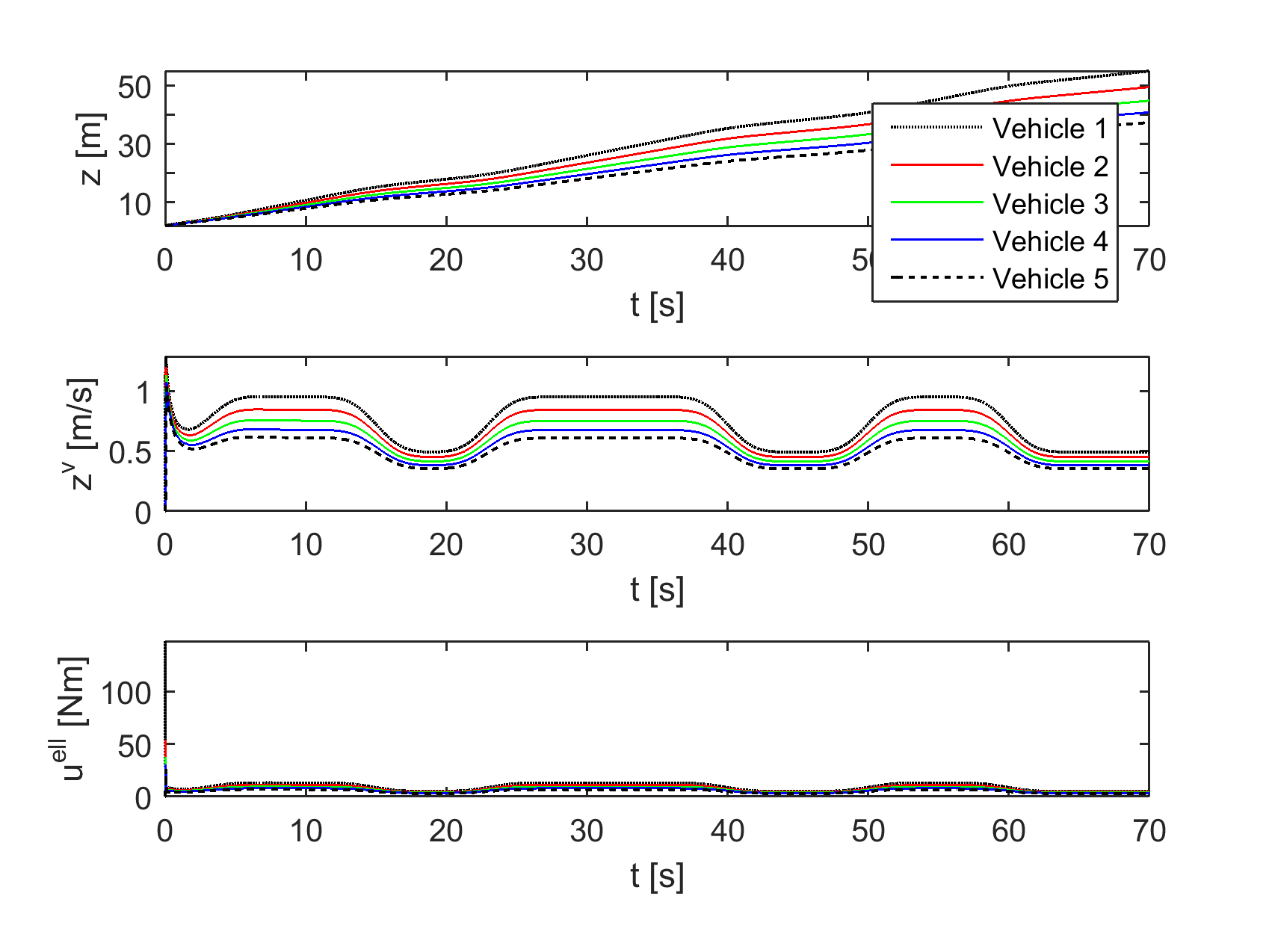}
\caption{Trajectories without using the predecessor's  control signal.}
\label{fig:nouprev}
\end{figure}

Next, we take heterogeneous vehicle dynamics: in order from the leader $k=0$ to vehicle $k=5$, $\gamma_k=0.003, 0.007, 0.011, 0.015, 0.019, 0.023$ and $\ell_k=0.3, 0.4, 0.45, 0.5, 0.6, 0.7$. First we keep the control form of the baseline experiment, without the term $-f_k(v_k) + f_{k-1}(v_k)$, to illustrate the need of compensating for vehicle heterogeneity. The results are shown in Figure \ref{fig:heterog}, top, where it is clearly seen that velocity agreement is lost in this case. If we then introduce the appropriate compensation term, the trajectories are those from Figure \ref{fig:heterog}, bottom, with nearly the same performance as in Figure \ref{fig:baseline}. Therefore, the controller efficiently compensates the heterogeneity. Note also the need to apply different control inputs to the vehicles due to their different dynamics.
\begin{figure}[!htb]
\centering
\includegraphics[width=\columnwidth]{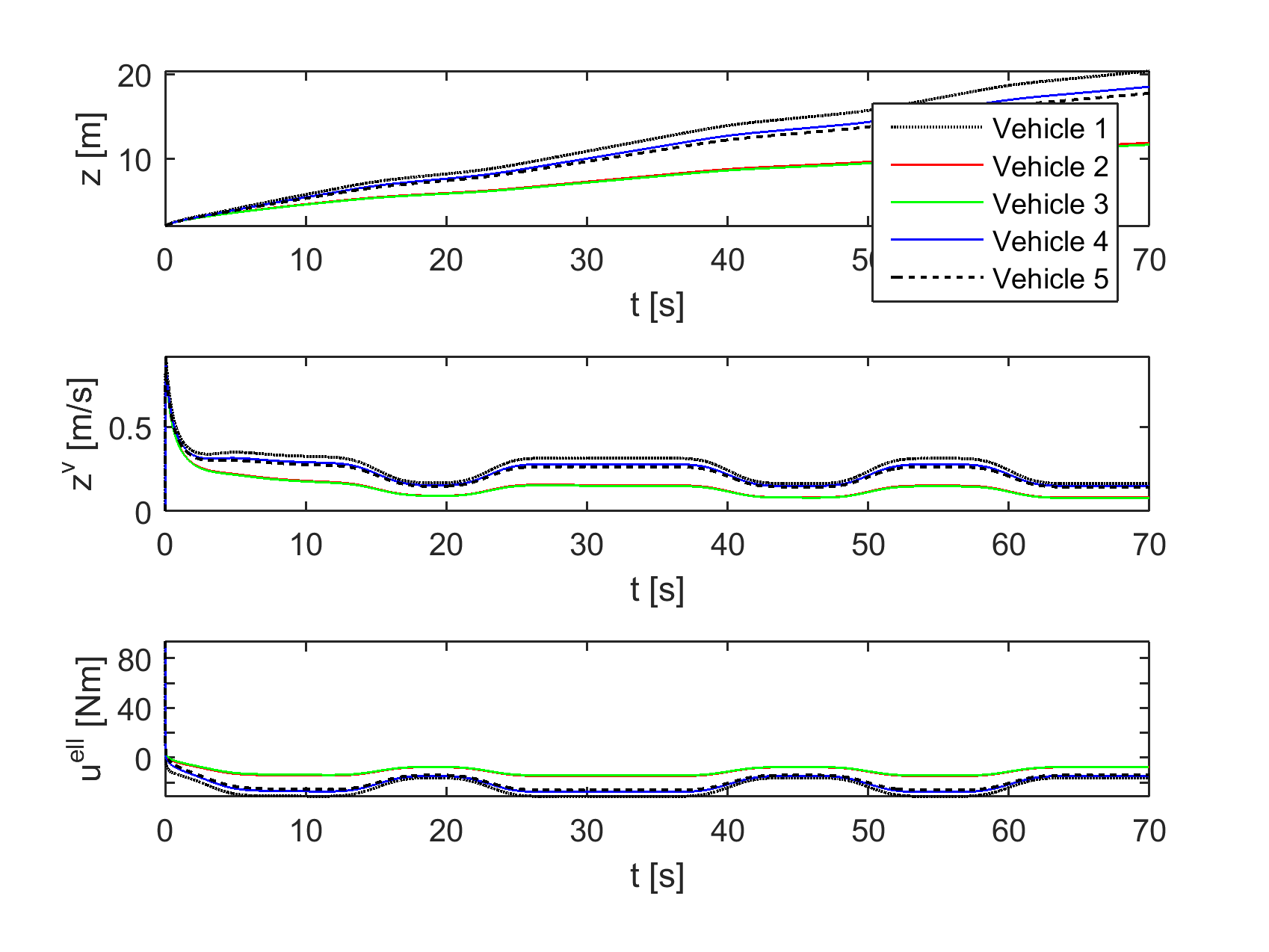}\\
\includegraphics[width=\columnwidth]{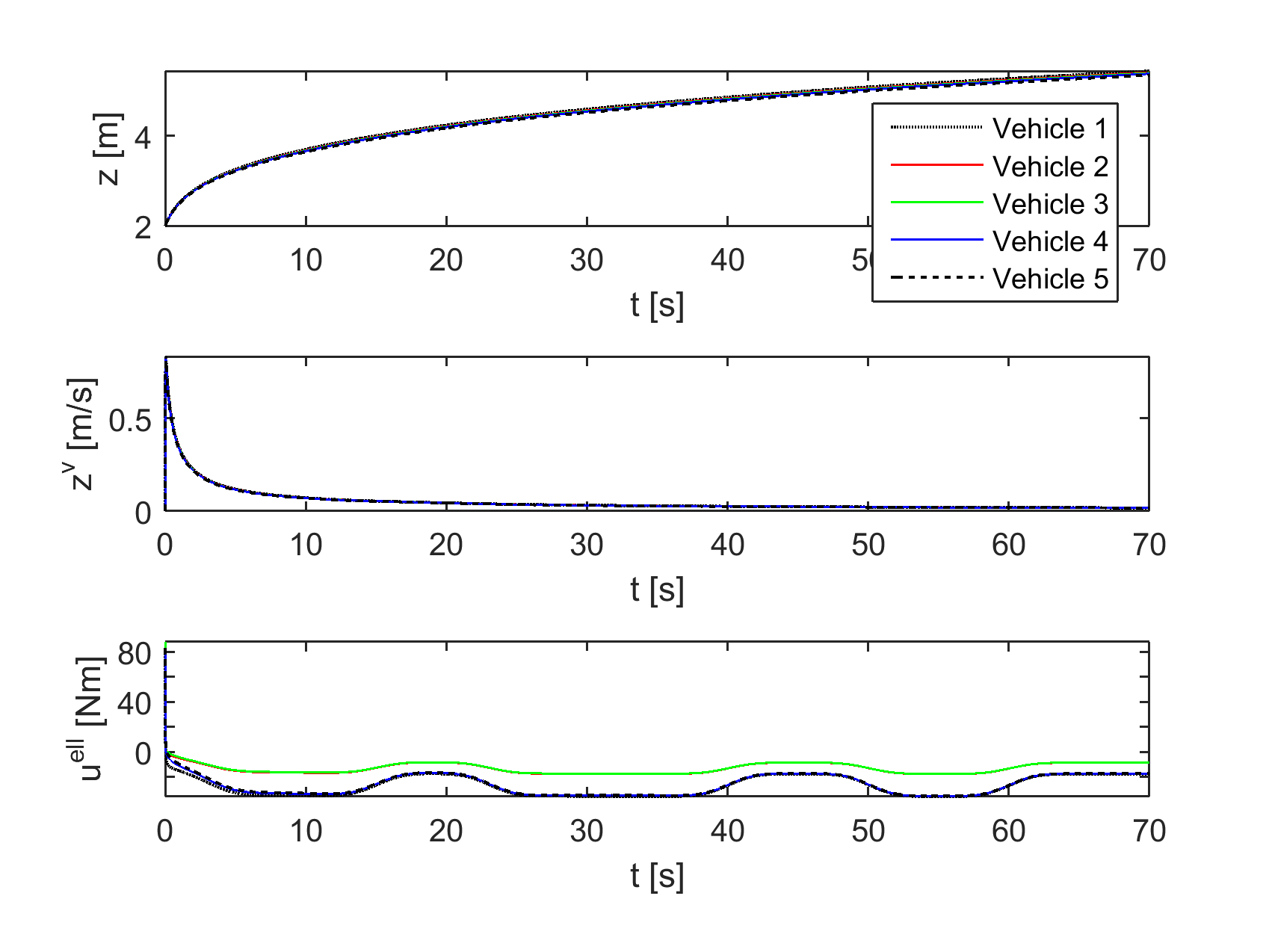}
\caption{Trajectories of the vehicles with heterogeneous dynamics. Top: uncompensated, bottom: compensated.}
\label{fig:heterog}
\end{figure}

On top of heterogeneity we now introduce a time delay of $\theta=0.2$\,s, which is quite conservative (about ten times the actual values) for digital radio communications and we apply the delay compensation mechanism from Section~\ref{delaycompensation}. The trajectories are those from Figure~\ref{fig:delay}. Note that the quantities reported are the actual instantaneous differences $z_k=y_{k-1}(t)-y_{k}(t)$, $z^v_k=v_{k-1}(t)-v_{k}(t)$, and not the actual regulated measurement (\ref{z-delayed})  from Section~\ref{delaycompensation}. 

\begin{figure}[!htb]
\centering
\includegraphics[width=\columnwidth]{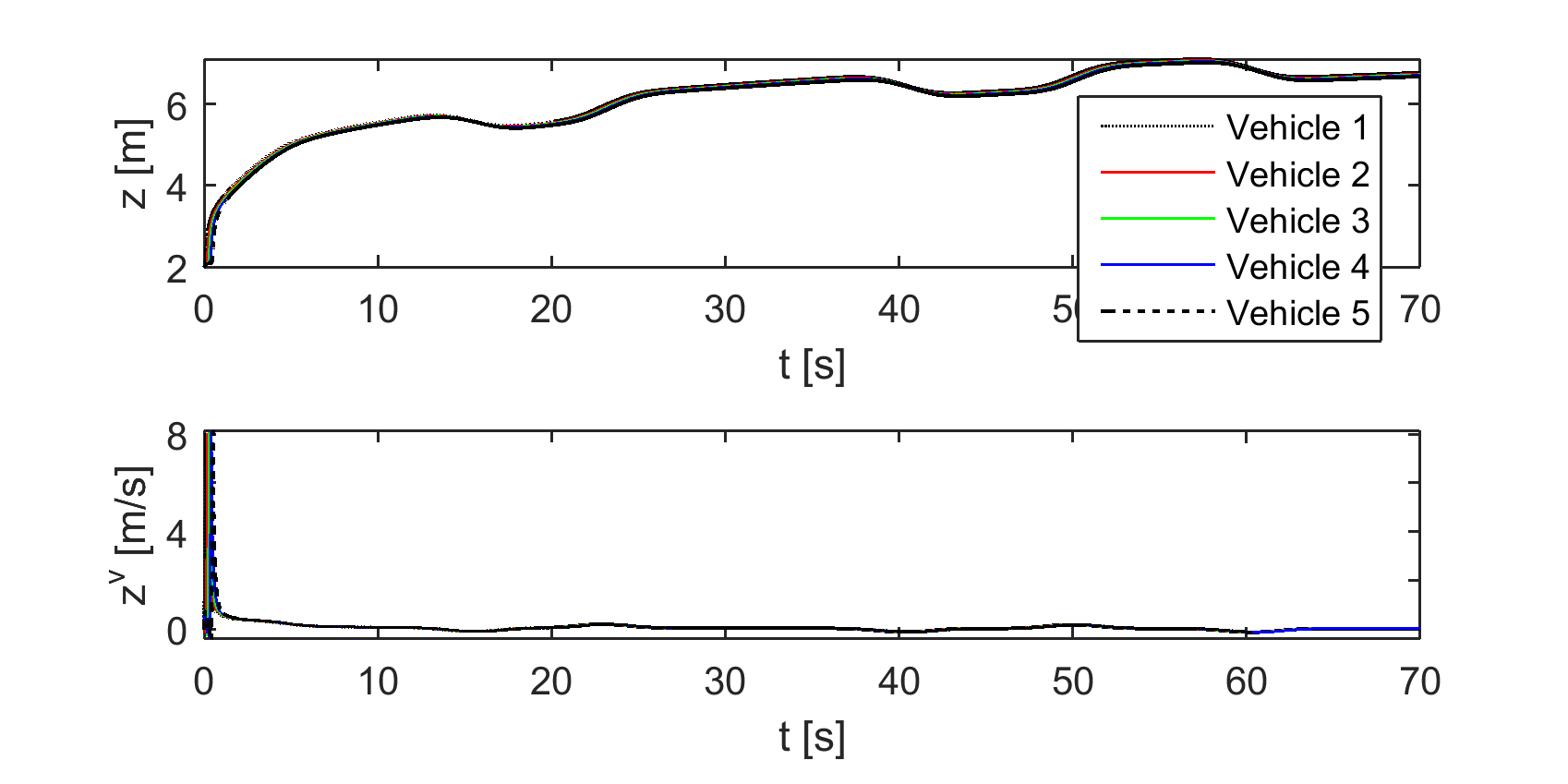}\\
\caption{Trajectories of vehicles with a relatively large $\theta=0.2$\,s time delay(control inputs are now shown in these graphs).}
\label{fig:delay}
\end{figure}

\subsection{A Heuristic for Optimal Control}

Since the stability analysis holds for any gain $\beta_k \geq \alpha_k$, a relevant practical problem is finding a good value for this gain. It turns out that this value depends on the particular objectives of the user. To illustrate, we choose several cost functions and optimize $\beta$ in the baseline experiment, for which $\forall k, \alpha_k = 55.56 =: \alpha$. Optimization is done by gridding the interval $[\alpha, 200]$ into $10$ values, setting all the follower gains to each value of $\beta$ in turn, and experimentally measuring the cost across all follower vehicles, integrated over time and averaged over the followers. The maximum of the range (here, $200$), is related to the maximum inputs that the vehicles can apply. The results are shown in \figref{fig:betaoptimization}. The first cost function penalizes velocity disagreements and the control effort: $g_1(z, z_v, u^\ell) = {(z^v)}^2 + {(u^\ell)}^2$ (we skip vehicle indices since the cost function is the same for al the vehicles). In this case, larger $\beta$ values are better, so we suggest taking the largest value that is achievable given the physical limits of the vehicles' drive train. However, if we introduce a ``safety'' premium  and add a `barrier' term $100/{\vert z \vert}$ to penalize relatively small interspacing distances, obtaining cost function $g_2$, the situation is reversed: the larger $\beta$ values focus too much on reducing velocity disagreements to the detriment of the interspacing distances, so in this case $\beta=\alpha$ works best. Depending on the particular weights, the optimal value may also be inside the interval, e.g. for $g_3(z, z_v, u^\ell) = 30/{\vert z \vert} + {(z^v)}^2 + {(u^\ell)}^2$, $\beta\approx88$ works best (keeping in mind the resolution of our grid is limited; to find a better approximation of the true optimum, a simple sample-based optimization method could be used, such as golden section rule or even binary search).

\begin{figure}[!htb]
\centering
\includegraphics[width=0.8\columnwidth]{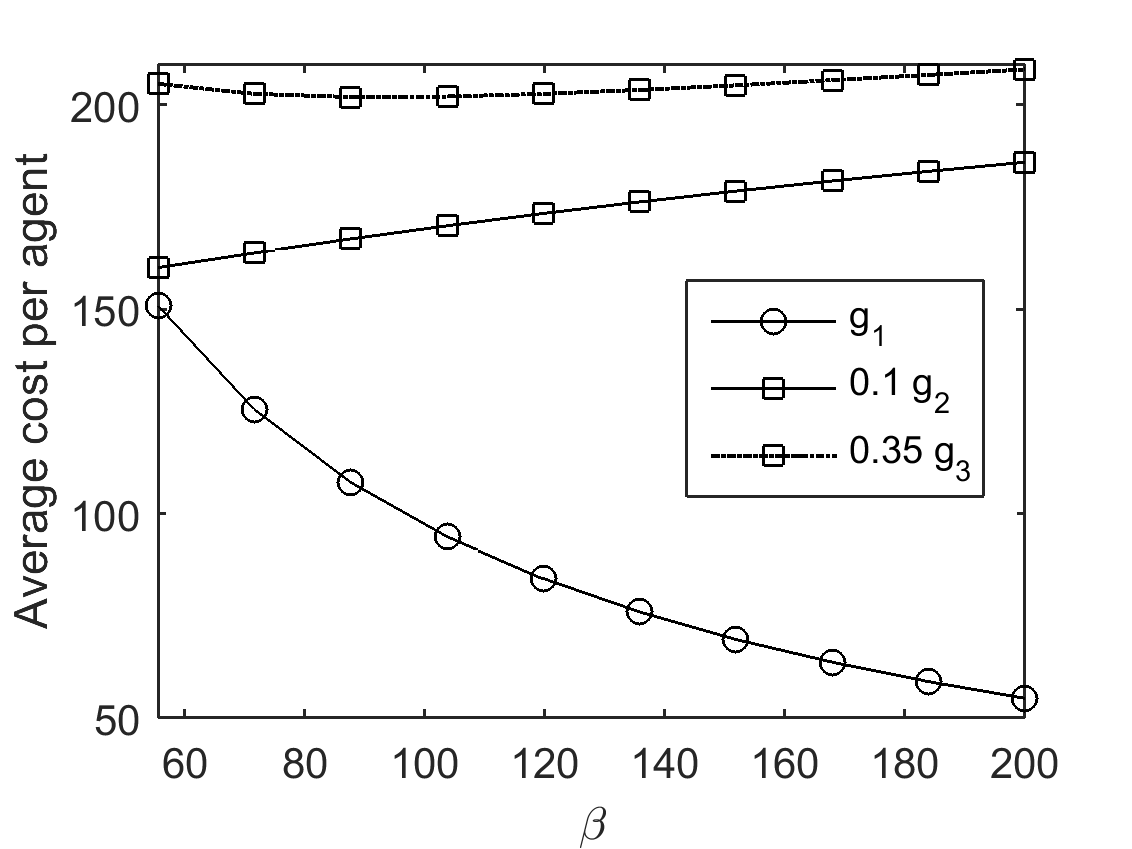}
\caption{Optimization of the gain $\beta$. To make the graph easier to read, some costs are multiplied by scaling factors, which leaves the optimal $\beta$ unchanged.}
\label{fig:betaoptimization}
\end{figure}

\section{Conclusions}

We have presented a novel method for the distributed control of  a string of heterogenous, nonlinear agents  guaranteeing collision avoidance and topology preservation in the presence of communication induced time-delays. Numerical experiments seem to suggest that the proposed scheme also benefits from a remarkable robustness to time-varying delays. The study of the root cause of this robustness is the scope of future investigations along with the study of more complicated cost functionals $\mathcal{J}(\cdot)$, introduced in Subsection~\ref{CDCA} and also the study of more general formation topologies, including those containing self-loops.

\bibliographystyle{IEEEtran}
\bibliography{plat}

\vspace{2mm}
\begin{IEEEbiography}[{\includegraphics[width=1in,height=1.25in,clip,keepaspectratio]{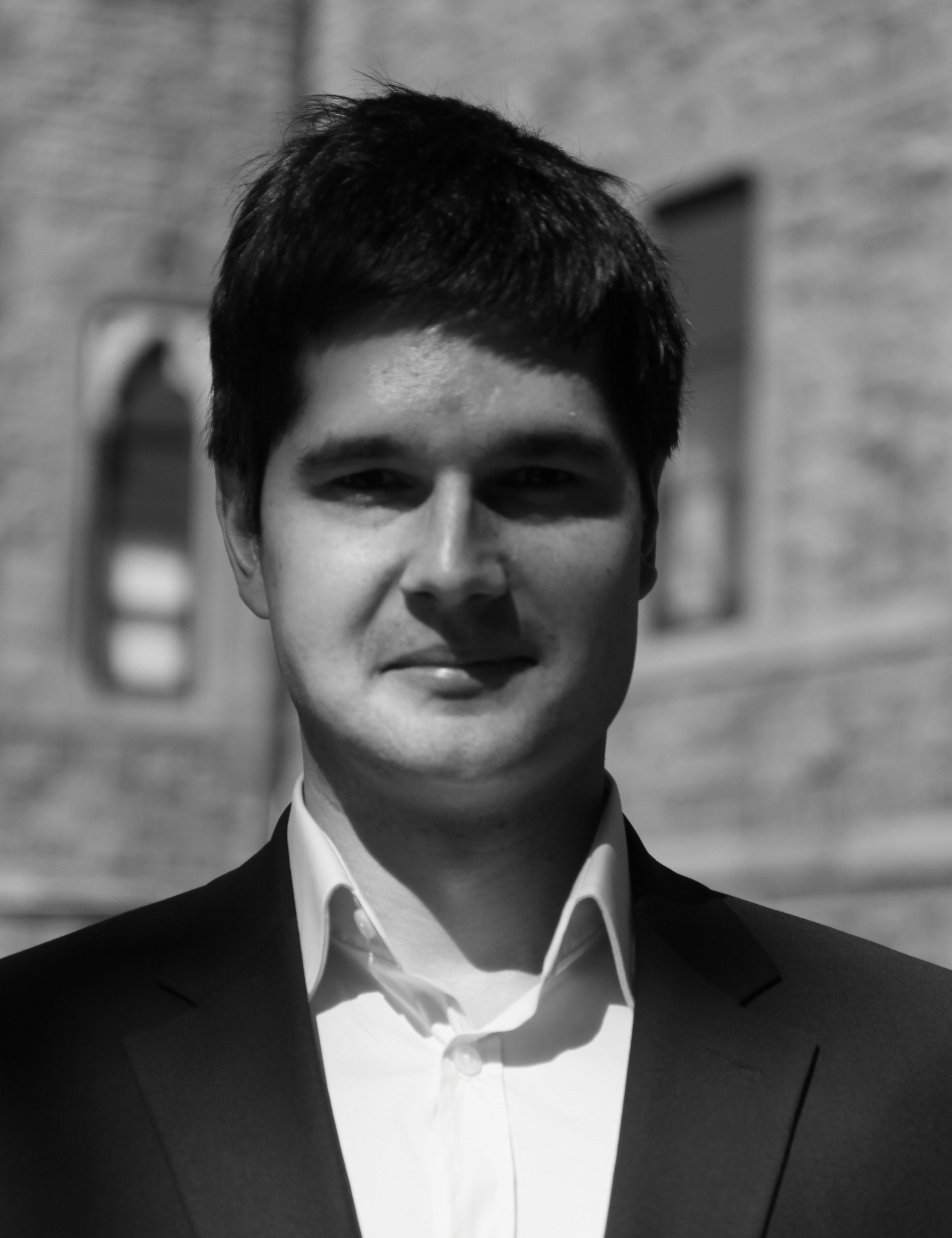}}]{\c{S}erban Sab\u{a}u}
\c{S}erban Sab\u{a}u received the M.S. degree in electrical engineering from ``Politehnica'' University Bucharest, Romania in 2002 and the Ph.D. degree in Electrical and Computer Engineering from the University of Maryland at College Park, in 2011. Before joining Stevens Institute of Technology in 2013 as an Assistant Professor, he was a postdoctoral researcher at the University of Pennsylvania. His research interests are in numerical algorithms for distributed control and distributed optimization. \end{IEEEbiography}

\begin{IEEEbiography}[{\includegraphics[width=1in,height=1.25in,clip,keepaspectratio]{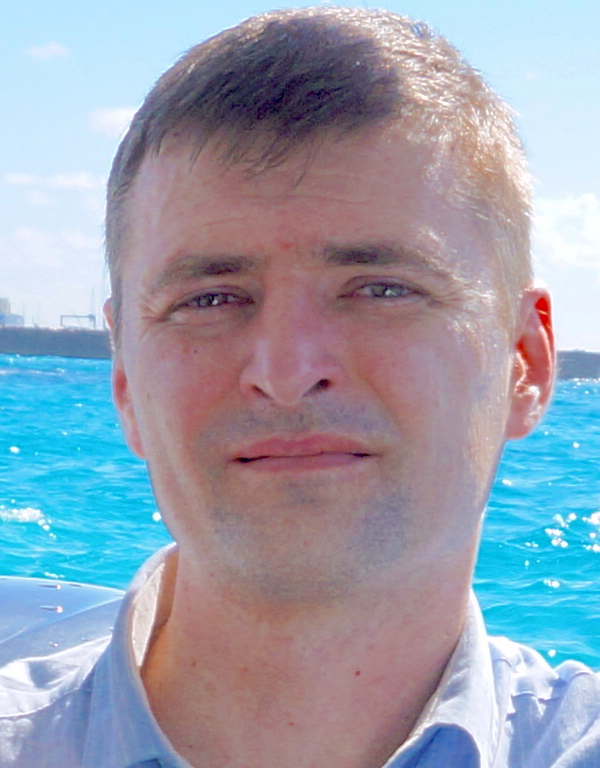}}]{Irinel-Constantin Mor\u{a}rescu}
is currently Associate Professor at Universit\'e de Lorraine and researcher at the Research Centre of Automatic Control (CRAN UMR 7039 CNRS) in Nancy, France. He received the B.S. and the M.S. degrees in Mathematics from University of Bucharest, Romania, in 1997 and 1999, respectively. In 2006 he received the Ph.D. degree in Mathematics and in Technology of Information and Systems from University of Bucharest and University of Technology of Compi\`egne, respectively. In November 2016 he received the "Habilitation \`a Diriger des Recherche" from Universit\'e de Lorraine. From March 2007 to December 2008 he was postdoctoral researcher at INRIA Grenoble, from January 2009 to December 2009 he was postdoctoral researcher at Jean Kuntzmann Laboratory and from January 2010 to October 2010 he was postdoctoral researcher at GipsaLab grenoble. His works concern stability and control of time-delay systems, stability and tracking for different classes of hybrid systems, consensus and synchronization problems.
\end{IEEEbiography}

\begin{IEEEbiography}[{\includegraphics[width=1in,height=1.25in,clip,keepaspectratio]{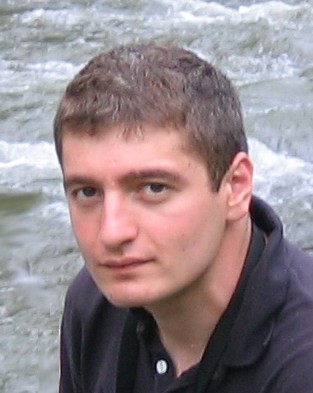}}]{Lucian Bu\c{s}oniu}
received the M.Sc. degree (valedictorian) from the Technical University of Cluj-Napoca, Romania, in 2003, and the Ph.D. degree (cum laude) from the Delft University of Technology, the Netherlands, in 2009. He is an associate professor with the Department of Automation at the Technical University of Cluj-Napoca, and has previously held research positions in the Netherlands and France. His research interests include planning, reinforcement learning, and aproximate dynamic programming for nonlinear optimal control; as well as multiagent systems and robotics. He received the 2009 Andrew P. Sage Award for the best paper in the IEEE Transactions on Systems, Man, and Cybernetics.
\end{IEEEbiography}

\begin{IEEEbiography}[{\includegraphics[width=1in,height=1.25in,clip,keepaspectratio]{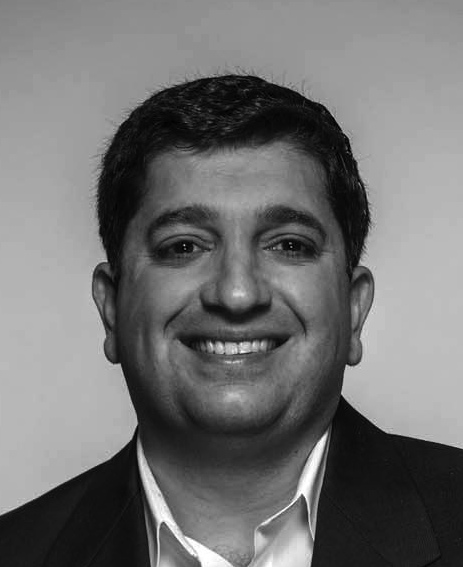}}]{Ali Jadbabaie}Ali Jadbabaie is the JR East Professor of Engineering and Associate Director of the Institute for Data, Systems and Society at MIT, where he is also on the faculty of the department of civil and environmental engineering and is a principal investigator in the Laboratory for Information and Decision Systems (LIDS). He is the director of the Sociotechnical Systems Research Center, one of MIT's 13 research laboratories and serves as the director of the Social and Engineering systems PhD Program. He received his Bachelors (with high honors) from Sharif University of Technology in Tehran, Iran, a Masters degree in electrical and computer engineering from the University of New Mexico, and his PhD in control and dynamical systems from the California Institute of Technology. He was a postdoctoral scholar at Yale University before joining the faculty at Penn in July 2002. Prior to joining MIT faculty, he was the Alfred Fitler Moore a Professor of Network Science and held secondary appointments in computer and information science and operations, information and decisions in the Wharton School. He was the inaugural editor-in-chief of IEEE Transactions on Network Science and Engineering, a new interdisciplinary journal sponsored by several IEEE societies. He is a recipient of a National Science Foundation Career Award, an Office of Naval Research Young Investigator Award, the O. Hugo Schuck Best Paper Award from the American Automatic Control Council, and the George S. Axelby Best Paper Award from the IEEE Control Systems Society. His students have been winners and finalists of student best paper awards at various ACC and CDC conferences. He is an IEEE fellow and a recipient of the 2016 Vannevar Bush Fellowship from the office of Secretary of Defense, and a member of the national Academic of Science, Engineering, and Medicine's Intelligence Science and Technology Expert Group (ISTEG). His current research interests are in distributed decision making, social learning,  multi-agent coordination and control, distributed optimization, network science, and network economics. \end{IEEEbiography}
\end{document}